\documentclass{article}

\PassOptionsToPackage{dvipsnames,table}{xcolor}

\usepackage{microtype}
\usepackage{makecell}
\usepackage{graphicx}
\usepackage{tabularx}
\usepackage{subfigure}
\usepackage{adjustbox}
\usepackage{multirow}
\usepackage{multicol}
\usepackage{booktabs} 
\usepackage{enumitem}

\usepackage{hyperref}

\usepackage{algorithm2e, algorithmic}

\renewcommand{\algorithmicrequire}{\textbf{Input:}}  
\renewcommand{\algorithmicensure}{\textbf{Output:}} 

\usepackage[accepted]{icml2025}

\usepackage{amsmath}
\usepackage{amssymb}
\usepackage{mathtools}
\usepackage{amsthm}
\usepackage{bbm}
\usepackage{caption}

\usepackage[capitalize,noabbrev]{cleveref}

\usepackage[most]{tcolorbox}

\tcbset{
  myalgoboxSkyBlue/.style={
    enhanced,
    box align=base,
    colback=SkyBlue!20,         
    colframe=SkyBlue,           
    boxrule=0.5pt,              
    arc=2mm,                    
    left=1.2mm, right=1.2mm, top=1.2mm, bottom=1.2mm,
    fonttitle=\bfseries,
    coltitle=black,
    attach title to upper,
    title style={boxrule=0pt, colback=SkyBlue!20}, 
  }
}
\tcbset{
  myalgoboxLavender/.style={
    enhanced,
    box align=base,
    colback=Lavender!20,         
    colframe=Lavender,           
    boxrule=0.5pt,              
    arc=2mm,                    
    left=1.2mm, right=1.2mm, top=1.2mm, bottom=1.2mm,
    fonttitle=\bfseries,
    coltitle=black,             
    attach title to upper,
    title style={boxrule=0pt, colback=Lavender!20} 
  }
}

\theoremstyle{plain}
\newtheorem{theorem}{Theorem}[section]
\newtheorem{proposition}[theorem]{Proposition}

\theoremstyle{definition}
\newtheorem{definition}[theorem]{Definition}

\theoremstyle{remark}
\newtheorem{remark}[theorem]{Remark}

\usepackage[textsize=tiny]{todonotes}

\icmltitlerunning{Multivariate Conformal Selection}

\begin{document}

\twocolumn[
\icmltitle{Multivariate Conformal Selection}




\icmlsetsymbol{equal}{*}

\begin{icmlauthorlist}
\icmlauthor{Tian Bai}{equal,math}
\icmlauthor{Yue Zhao}{ymath}
\icmlauthor{Xiang Yu}{merck}
\icmlauthor{Archer Y. Yang}{equal,math,mila}
\end{icmlauthorlist}

\icmlaffiliation{math}{Department of Mathematics and Statistics, McGill University, Montreal, Canada}
\icmlaffiliation{mila}{Mila - Quebec AI Institute, Montreal, Quebec, Canada}
\icmlaffiliation{ymath}{Department of Mathematics, University of York, York, UK}
\icmlaffiliation{merck}{MRL, Merck \& Co., Inc., Rahway, NJ, USA}

\icmlcorrespondingauthor{Archer Y. Yang}{archer.yang@mcgill.ca}

\icmlkeywords{Uncertainty Quantification, Machine Learning, Model-free Selective Inference, Conformal Inference, Multivariate Responses}

\vskip 0.3in
]



\printAffiliationsAndNotice{\icmlEqualContribution} 

\begin{abstract}
Selecting high-quality candidates from large datasets is critical in applications such as drug discovery, precision medicine, and alignment of large language models (LLMs). While Conformal Selection (CS) provides rigorous uncertainty quantification, it is limited to univariate responses and scalar criteria. To address this issue, we propose Multivariate Conformal Selection (mCS), a generalization of CS designed for multivariate response settings. Our method introduces regional monotonicity and employs multivariate nonconformity scores to construct conformal $p$-values, enabling finite-sample False Discovery Rate (FDR) control. We present two variants: \texttt{mCS-dist}, using distance-based scores, and \texttt{mCS-learn}, which learns optimal scores via differentiable optimization. Experiments on simulated and real-world datasets demonstrate that mCS significantly improves selection power while maintaining FDR control, establishing it as a robust framework for multivariate selection tasks.
\end{abstract}

\section{Introduction}

Selecting a subset of promising candidates from a large pool is crucial across various scientific and real-world applications. In drug discovery, researchers search vast chemical spaces to identify compounds with strong effects, such as high binding affinity to a specific target \citep{szymanski2011adaptation, scannell2022predictive, sheridan2015ecounterscreening, zhang2025artificial}. Similarly, precision medicine aims to identify positive individual treatment effects~\cite{lei2021conformal}, and post-hoc certification of large language model (LLM) outputs seeks to retain only trustworthy generations that meet user-defined criteria \citep{gui2024conformal}.
 
In these settings, true test responses (e.g., binding affinity or alignment score) are often unavailable, requiring selection to rely on machine learning model predictions. Since selected targets drive downstream decisions, quantifying selection uncertainty is essential for maintaining efficiency. Conformal selection~\citep{jin2023selection} provides a model-agnostic framework for selection with uncertainty quantification by extending conformal prediction~\citep{Vovk2005} to multiple hypothesis testing, using conformal $p$-values~\citep{Bates2023} and multiple testing corrections. It has shown promise in real-world drug discovery and LLM alignment~\citep{gui2024conformal}. 

However, existing conformal selection methods are limited to univariate responses and selection targets of the form $y > c$, where $c$ is a user-defined threshold. Many real-world applications require selection based on multiple interdependent criteria. For instance, LLM outputs must simultaneously satisfy alignment requirements such as fairness, safety, and correctness~\citep{bai2022training}, which are better represented as multivariate alignment scores. This highlights the need for a principled selection method with uncertainty quantification in multivariate settings.

In this paper, we extend the conformal selection framework to multivariate response settings. 
To ensure finite-sample \textit{false discovery rate} (FDR) control in our procedure, we generalize the concept of monotonicity~\citep{jin2023selection} to \textit{regional monotonicity} for the nonconformity function. We propose two types of nonconformity scores that satisfy this property: (i) distance-based nonconformity scores for regular-shaped and convex target regions, and (ii) a learning-based method for optimizing nonconformity scores. The latter approach leverages a loss function that penalizes either the smooth selection size or the conformal $p$-value to learn an optimal score. This method is particularly effective when the dimension of responses is large or when the target region is irregular or nonconvex. Through experiments on simulated and real-world datasets, both variants of mCS demonstrates enhanced selection power over baseline methods while ensuring finite-sample FDR control.

\section{Background and Related Work} \label{sec:background}

\paragraph{Problem Setup}
We let $\boldsymbol{x}\in \mathbb{R}^{p}$ represent the $p$-dimensional features, and let $\boldsymbol{y}\in \mathbb{R}^{d}$ denote the $d$-dimensional multivariate response variables. 
We consider a training dataset $\mathcal{D}_{\textnormal{train}} = \{\boldsymbol{x}_i, \boldsymbol{y}_i\}_{i=1}^n$ and a test dataset $\mathcal{D}_{\textnormal{test}} = \{\boldsymbol{x}_{n+j}\}_{j=1}^m$, where the  corresponding test responses $\{\boldsymbol{y}_{n+j}\}_{j=1}^m$ are unobserved. 
We further assume that the combined set of training and test samples $\{\boldsymbol{x}_i, \boldsymbol{y}_{i}\}_{i=1}^{n+m}$ are drawn i.i.d.\footnote{Later, we will relax the i.i.d. assumption to exchangeability conditions.}\,from an unknown, arbitrary distribution  $\mathcal{D}_{X\times Y}$.

We formulate the selection problem as follows:
Given a predefined $d$-dimensional closed region $R \subseteq \mathbb{R}^{d}$, our goal is to identify a subset of indices $\mathcal{S} \subseteq \{1, \dots, m\}$ from $\mathcal{D}_{\textnormal{test}}$ such that as many test observations $j \in \mathcal{S}$ satisfy $\boldsymbol{y}_{n+j} \in R$ as possible, while controlling the FDR \citep{bh1995} below a user-specified level $q$. 
The FDR is defined as the expected proportion of false discoveries ($j \in \mathcal{S}$ but $\boldsymbol{y}_{n+j} \notin R)$ among all selected observations
\begin{align} \label{eq:fdr}
\text{FDR} = \mathbb{E}\Bigg[\frac{|\mathcal{S} \cap \mathcal{H}_0|}{|\mathcal{S}|} \Bigg]
\end{align}
where $\mathcal{H}_0 = \{j: \boldsymbol{y}_{n+j} \notin R\}$, with the convention that $0/0=0$ in the fraction above. This criterion measures the overall Type-I error rate of the selection procedure.

The overall Type-II error of selection can be quantified by the \emph{power}, defined as the expected proportion of desirable observations $(\boldsymbol{y}_{n+j} \in R)$ that are correctly selected,
\begin{align} \label{eq:power}
\text{Power} = \mathbb{E}\Bigg[\frac{|\mathcal{S} \cap \mathcal{H}_1|}{|\mathcal{H}_1|} \Bigg]
\end{align}
where $\mathcal{H}_1 = \{j: \boldsymbol{y}_{n+j} \in R\}$.
The Type-II error of selection is therefore $(1 - \text{Power})$. An ideal selection procedure should aim to maximize the power while keeping the FDR below the specified nominal level.

\paragraph{Conformal Prediction}


Conformal prediction (CP) \cite{Vovk2005} is a popular framework for uncertainty quantification that constructs prediction intervals on a per-sample basis.
Assuming exchangeable calibration and test data, CP provides prediction sets $\widehat{C}_{1-\alpha}(\boldsymbol{x})$ with finite-sample coverage guarantees for $\alpha\in (0,1)$:
$$
\mathbb{P}(y \in \widehat{C}_{1-\alpha}(\boldsymbol{x})) \geq 1 - \alpha.
$$
Although CP was originally designed for univariate responses, numerous studies have proposed multivariate generalizations \citep{kuleshov2018conformal, bates2021distribution, messoudi2022ellipsoidal, johnstone2021conformal, feldman2023calibrated, park2024semiparametric, klein2025multivariate}.
However, these multivariate CP methods are not directly applicable to our selection problem. The primary objective of CP -- constructing confidence sets for predictions -- does not naturally align with selection tasks.
Specifically, the potentially complex shapes of the multivariate CP sets $\widehat{C}_{1-\alpha}$ may be incompatible with the pre-defined target region $R$.

Even in the simpler cases, such as when the response is binary, using CP for selection introduces a multiplicity issue \cite{jin2023selection}. 
In this context, CP only controls the per-comparison error rate (PCER), which differs from the FDR. PCER also measures the Type-I error, and is defined as~\eqref{eq:fdr} with the denominator replaced by $m$, the size of the test dataset.
By definition, the PCER is always smaller than the FDR, making it a less stringent error control criterion. As a result, procedures controlling PCER may fail to meet the stricter requirements of FDR control.

\paragraph{Conformal Selection}
Conformal selection (CS) \cite{jin2023selection} is a model-agnostic selection framework that guarantees  finite-sample FDR control. However, CS only considers the univariate response case ($d=1$) and assumes that the selection region takes the form $[c, +\infty)$, where $c$ is a predefined threshold.  In this setting, CS formulates one hypothesis test per candidate:
\begin{align*}
    H_{0j}: y_{n+j} \leq c \quad \text{vs.} \quad
    H_{1j}: y_{n+j} > c .
\end{align*}
Rejecting the null hypothesis $H_{0j}$ indicates that the $j$-th test sample is selected, as its response is deemed to exceed the threshold $c$.

CS uses nonconformity scores to guide its selection process. A nonconformity measure quantifies how atypical (or nonconforming) an observation is, based on the relationship between inputs and responses. 
For calibration samples, the nonconformity scores are
$V_i = V(\boldsymbol{x}_i, y_i)$, and for test samples, they are $\widehat{V}_{n+j} = V(\boldsymbol{x}_{n+j}, c)$, where $c$ replaces the unobserved $y_{n+j}$. These scores are then used to compute conformal $p$-values through a rank-based comparison of $\widehat{V}_{n+j}$ against the calibration scores $V_1, \dots, V_n$. A lower rank of $\widehat{V}_{n+j}$ relative to the calibration scores provides stronger evidence for rejecting $H_{0j}$.

To determine the final selected subset $\mathcal{S}$, CS applies the Benjamini-Hochberg (BH) procedure~\citep{bh1995, bh1997}, a widely used method for controlling FDR in multiple testing setting, to the set of conformal $p$-values. The use of the BH procedure ensures that the overall Type-I error rate is kept below the specified level.

\section{Multivariate Conformal Selection}

In this section, we introduce the key concepts, procedures, and theoretical foundations of multivariate Conformal Selection (mCS).

First, for each $d$-dimensional multivariate response $\boldsymbol{y}_{n+j}$, mCS performs the following hypothesis test:
\begin{align*}
    H_{0j}: \boldsymbol{y}_{n+j} \in R^c \quad \text{vs.} \quad H_{1j}: \boldsymbol{y}_{n+j} \in R,
\end{align*}
where $R$ represents an arbitrary pre-defined closed target region in $\mathbb{R}^{d}$.  Multivariate responses $\boldsymbol{y}_{n+j}$ can represent either regression or classification outcomes. In this paper, we focus on the more challenging regression setting. For a discussion of the classification case, please see Appendix~\ref{sec:cls_discussion}. The mCS process consists of three main steps:
\begin{enumerate}[itemsep=0.3pt, topsep=1pt]
   \item \textbf{Training:} Construct a multivariate predictive model $\hat{\mu}$ for $\boldsymbol{y}$. This model can be obtained using any suitable machine learning algorithm.
   \item \textbf{Calibration:} Build a regionally monotone multivariate nonconformity function based on $\hat\mu$, and evaluate this function on the calibration dataset and test dataset. Subsequently, we compute the conformal $p$-values for each test sample.
   \item \textbf{Thresholding:} Apply the Benjamini-Hochberg (BH) procedure as in the original CS procedure to the set of conformal $p$-values, yielding the final selection set $\mathcal{S}$.
\end{enumerate}

Both the training step and the calibration step rely on the labeled dataset $\mathcal{D}_{\textnormal{train}}$.
In the case where a model $\hat\mu$ is already available, mCS can be directly applied using all labeled data for calibration. Otherwise, the training data is divided into two subsets: one for model training (the proper training dataset) and the other for calibration. For simplicity, in the following discussion, we assume that $\hat\mu$ is already available and all training data $\{\boldsymbol{x}_i, \boldsymbol{y}_i\}_{i=1}^n$ are used for calibration so that $\mathcal{D}_{\textnormal{cal}}= \mathcal{D}_{\textnormal{train}}$.

The conformal $p$-values~\citep{Bates2023} are used to perform the hypothesis tests. If the true responses $\{\boldsymbol{y}_{n+j}\}_{j=1}^m$ were observed, the \textit{oracle} conformal $p$-value would be defined as
\begin{align}
\label{eq:oracle_p}
    p_j^* = \frac{\sum_{i=1}^n\mathbbm{1}\{V_i < {V}_{n+j} \} + U_j (1 + \sum_{i=1}^n \mathbbm{1}\{V_i = {V}_{n+j}\}) }{n+1}
\end{align}
where $V_i = V(\boldsymbol{x}_{i}, \boldsymbol{y}_{i})$ for $i=1,\dots,n+m$ and $V$ is a multivariate nonconformity function based on $\hat{\mu}$, and $U_j \sim \text{Unif}(0,1)$ is an independent random variable for the tie-breaking of the nonconformity scores. 
We defer the specific choice of $V$ to later sections. 

The evaluation of the oracle $p$-value $p_j^*$ is infeasible, because that in the above definition~\eqref{eq:oracle_p}, the computation of $V_{n+j} = V(\boldsymbol{x}_{n+j}, \boldsymbol{y}_{n+j})$ requires knowledge of the unobserved response $\boldsymbol{y}_{n+j}$. 
To address this issue, we replace $V_{n+j}$ with $\widehat{V}_{n+j} = V(\boldsymbol{x}_{n+j}, \boldsymbol{r}_{n+j})$, where $\boldsymbol{r}_{n+j}$ is an arbitrarily chosen point in the target region $R$, yielding the (practical) conformal $p$-values:
\begin{align}
\label{eq:conf_p}
    p_j = \frac{\sum_{i=1}^n\mathbbm{1}\{V_i < \widehat{V}_{n+j} \} + U_j  (1 + \sum_{i=1}^n \mathbbm{1}\{V_i = \widehat{V}_{n+j}\}) }{n+1} .
\end{align}
Standard results from conformal inference ensure that the oracle conformal $p$-values $p_j^*$ follow the $\text{Unif}(0,1)$ distribution, in particular implying that they are conservative in the sense that $p_j^*$ as a random variable has a super-uniform distribution on $[0,1]$, satisfying $\mathbb{P}(p_j^* \leq \alpha) \leq \alpha$ ~\citep{Bates2023}.  
To guarantee that the practical conformal $p$-values $p_j$ also maintain this conservativeness, 
the nonconformity function $V$ must satisfy a property that we introduce as \textit{regional monotonicity}. With this property the $p_j$'s in the resulting mCS procedure will ensure the control of the FDR.
We formally define the regional monotonicity as follows:
\begin{definition}[Regional Monotonicity]
\label{def:monotone}
    A nonconformity score $V: \mathcal{X} \times \mathcal{Y} \rightarrow \mathbb{R}$ satisfy the regional monotone property if $V(\boldsymbol{x}, \boldsymbol{y}') \leq V(\boldsymbol{x}, \boldsymbol{y})$ for any $\boldsymbol{x} \in \mathcal{X}$, $\boldsymbol{y}' \in R^c$ and $\boldsymbol{y} \in R$.
\end{definition}

Definition~\ref{def:monotone} leads to the conservativeness of $p_j$. The following proposition formalizes this result, with proof available in Appendix~\ref{proof:conservative}.

\begin{proposition} \label{prop:conservative}
    Given that the calibration data $\{(\boldsymbol{x}_i, \boldsymbol{y}_i)\}_{i=1}^n$ together with the $j$-th data test data point $(\boldsymbol{x}_{n+j}, \boldsymbol{y}_{n+j})$ are exchangeable for $j \in \{1, \dots, m\}$,  regionally monotone nonconformity scores $V$  ensures that the conformal $p$-value $p_j$ defined in \eqref{eq:conf_p} is conservative in the following sense, 
    \begin{align}
\label{eq:conservative}
        \mathbb{P}(p_j \leq \alpha \textnormal{ and } j\in \mathcal{H}_0) \leq \alpha, \quad \textnormal{for all } \alpha \in (0,1).
    \end{align}
\end{proposition}

\begin{remark}[Clarification on conservativeness]
    The conservativeness described in~\eqref{eq:conservative} differs from the conventional notion of statistical conservativeness, which is not conditional on the event $j \in \mathcal{H}_0$. Due to the inherent randomness in $\boldsymbol{y}_{n+j}$ within our hypothesis tests, an unknown dependency exists between the $p$-values and the event ${ j \in \mathcal{H}_0 }$. Consequently, the standard form of conservativeness does not hold in this context.
\end{remark}

\begin{remark}[Univariate monotonicity as a special case]
    Regional monotonicity generalizes the univariate monotonicity concept introduced in CS \cite{jin2023selection}, which was originally defined for nonconformity scores increasing in their second argument. The original definition is restricted to the univariate case, where the target region $R$ is specified as $(c, +\infty)$. For a univariate nonconformity score $V: \mathcal{X} \times \mathcal{Y} \to \mathbb{R}$, monotonicity implies regional monotonicity. Specifically, for any $y \in R^c = (-\infty, c]$ and $y' \in R$, it holds that $V(\boldsymbol{x}, y) \leq V(\boldsymbol{x}, y')$ for all $\boldsymbol{x} \in \mathcal{X}$.
    Moreover, even in the univariate case, the original definition of monotonicity across the entire domain $\mathcal{X}$ is unnecessary and monotonicity across the regions $R^c$ and $R$ suffices.
\end{remark}

Once a regionally monotone multivariate nonconformity score is defined, conformal $p$-values can be computed using~\eqref{eq:conf_p}. Leveraging these conformal $p$-values, mCS again applies the Benjamini-Hochberg procedure~\citep{bh1995} to construct a selection set $\mathcal{S}$. The complete approach is outlined in Algorithm~\ref{algo:mcs}.

        




\begin{algorithm}[H]
\caption{mCS: Multivariate Conformal Selection}
\label{algo:mcs}
\begin{algorithmic}[1]
    \REQUIRE{Calibration data $\mathcal{D}_{\textnormal{cal}}=\{(\boldsymbol{x}_i, \boldsymbol{y}_i)\}_{i=1}^n$, test data $\mathcal{D}_{\textnormal{test}}=\{\boldsymbol{x}_{n+j}\}_{j=1}^m$, target target region $R$, FDR nominal level $q \in (0,1)$, regionally monotone nonconformity score $V: \mathcal{X} \times \mathcal{Y} \rightarrow \mathbb{R}$.}
    \ENSURE{Selection set $\mathcal{S}$.}
    \STATE Compute $V_i = V(\boldsymbol{x}_i,\boldsymbol{y}_i)$ for $i = 1, \dots, n$, and $\widehat{V}_{n+j} = V(\boldsymbol{x}_{n+j}, \boldsymbol{r}_{n+j})$ for $j = 1, \dots, m$ with $\boldsymbol{r}_{n+j} \in R$.
    \STATE Construct conformal $p$-values $\{p_i\}_{i=1}^m$ as in~\eqref{eq:conf_p}.
    \STATE (BH procedure) Compute the BH selection threshold $k^* = \max\{ k \in \mathbb{Z}_{\geq 0}: \sum_{j=1}^m \mathbbm{1}\{p_j \leq qk/m\} \geq k \}$, and return the selection set as $\mathcal{S} = \{j: p_j \leq qk^*/m\}$.
\end{algorithmic}
\end{algorithm}

The following theorem shows that Algorithm~\ref{algo:mcs} controls FDR in finite sample. The proof can be found in Appendix~\ref{proof:main}.

\begin{theorem}
\label{thm:main}
    Suppose $V$ is a regionally monotone nonconformity score, and for any $j \in \{1, \dots, m\}$, the random variables $V_1, \dots, V_n, V_{n+j}$ are exchangeable conditioned on $\{\widehat{V}_{n+\ell}\}_{\ell \neq j}$. Then, for any $q \in (0,1)$, the output $\mathcal{S}$ of mCS satisfies FDR $\leq q$.
\end{theorem}


\section{Choices of Nonconformity Score} \label{sec:choice_of_score}

While the previous sections demonstrate that FDR-controlled selection can be performed using any regionally monotone nonconformity score, the selection power heavily depends on the quality of the chosen score. 
Although related studies have explored this issue in the context of CP~\citep{romano2019conformalized, kivaranovic2020adaptive, sesia2020comparison},  there has been limited focus on the choice of scores specific for CS.

In this section, we introduce two types of nonconformity score that satisfies the regional monotonicity. As a result of Theorem~\ref{thm:main}, applying Algorithm \ref{algo:mcs} with the purposed scores would guarantee FDR control.

\subsection{\texttt{mCS-dist}: Distance-based Scores} \label{sec:dist}

For multivariate selection, we propose distance-based nonconformity scores of the following form:
\begin{align} \label{eq:dist_score}
    V(\boldsymbol{x}, \boldsymbol{y}) = D_1(\boldsymbol{y}, R^c) - D_2(\hat\mu(\boldsymbol{x}), R^c)
\end{align}
where $\hat\mu$ is a trained predictive model of $\boldsymbol{y}$, and $D_1$ and $D_2$ are distance functions. 
Two examples of distance-based nonconformity score are as follows:
\begin{align}
    1. (\text{regular}) \quad D_1(\boldsymbol{z}, R^c) &= D_2(\boldsymbol{z}, R^c) = \inf_{\boldsymbol{s} \in R^c} \|\boldsymbol{z} - \boldsymbol{s}\|_p. \label{eq:dist_score_1}\\
    2. (\text{clipped}) \quad D_1(\boldsymbol{z}, R^c) &= M \cdot \mathbbm{1}\{\boldsymbol{z} \notin R^c \cup \partial R\}, \nonumber \\
    D_2(\boldsymbol{z}, R^c) &= \inf_{\boldsymbol{s} \in R^c} \|\boldsymbol{z} - \boldsymbol{s}\|_p , \label{eq:dist_score_2}
\end{align}
where $M$ is a large constant that serves as a relaxation of infinity. We discuss the role of $M$ in the following paragraphs.
These scores generalize the \textit{signed error} score and the \textit{clipped} score~\citep{jin2023selection}, respectively.

This formulation enables the decomposition of $V(\boldsymbol{x}, \boldsymbol{y})$ into two terms.
The first term $D_1$ inherently ensures the regional monotonicity of the score $V$ (Definition~\ref{def:monotone}), as $D_1$ is a distance function satisfying $0 = D_1(\boldsymbol{y}', R^c) \leq D_1(\boldsymbol{y}, R^c)$ for any $\boldsymbol{y}' \in R^c$ and $\boldsymbol{y} \in R$. 
Meanwhile, the second term $D_2$ measures the distance between the predicted responses of the points and $R^c$; this term is designed to increase as $\hat\mu(\boldsymbol{x})$ moves away from $R^c$, ensuring that the test points with the predicted responses having large distance from $R^c$ will yield smaller test scores $\widehat{V}_{n+j}$. According to~\eqref{eq:conf_p}, these points will have smaller conformal $p$-values and are more likely to be rejected (selected) by the BH procedure. Note that when the predictive model $\hat\mu$ outputs an estimated conditional distribution $\widehat{P}(\boldsymbol{y}| \boldsymbol{x})$ -- as in classification, conditional density estimation, or Bayesian models -- the second term $D_2$ can be replaced by the predicted probability of being in the target region: $\boldsymbol{y} \in R$, i.e. $\int \mathbbm{1}\{
\boldsymbol{y}\in R\} \mathrm{d}\widehat{P}(\boldsymbol{y} | \boldsymbol{x})$. This serves the same purpose: points with high predicted probability of satisfying the selection criterion will receive lower scores and are more likely to be selected.

We note that computing $\widehat{V}_{n+j}=V(\boldsymbol{x}_{n+j}, \boldsymbol{r}_{n+j})$ in \eqref{eq:conf_p} requires choosing a point $\boldsymbol{r}_{n+j} \in R$ to ensure FDR control.
Although any point in  $R$ works, selecting $\boldsymbol{r}_{n+j}$ on the boundary  $\partial R$ is optimal for maximizing selection power for both regular and clipped scores. In this case, for any fixed $\boldsymbol{x}_{n+j}$, the test score $\widehat{V}_{n+j}$ achieves its minimum, ensuring it to be smaller than a larger proportion of  calibration scores ${V}_{i}=V(\boldsymbol{x}_{i}, \boldsymbol{y}_{i})$. For example, with the clipped score, if $\boldsymbol{r}_{n+j}\in \partial R$, then the first term $D_1$ becomes 0 in $\widehat{V}_{n+j}$, while for any calibration samples  with $\boldsymbol{y}_i \in R$, we have $D_1=M$ 
(except when $\boldsymbol{y}_i$ lies exactly on $\partial R$, which occurs with zero probability.) This yields smaller $p$-values $p_j$ and enables more samples to be selected.
We require $M$ to be large for the same reason.

The following asymptotic analysis provides a theoretical basis for preferring the second, clipped score~\eqref{eq:dist_score_2} over the first score~\eqref{eq:dist_score_1} when choosing a distance-based nonconformity score. 
This result extends the original CS framework~\citep{jin2023selection} to the multivariate setting, as formalized in the following theorem:

\begin{theorem}
\label{thm:heuristic}
    Let $V$ be any fixed regionally monotone nonconformity score, and suppose $\{(\boldsymbol{x}_i, \boldsymbol{y}_i)\}_{i=1}^{n+m}$ are exchangeable from distribution  $\mathcal{D}_{X\times Y}$. Let $(\boldsymbol{x}, \boldsymbol{y})$ denote a random pair also drawn from $\mathcal{D}_{X\times Y}$, and define $F(v, u) = \mathbb{P}(V(\boldsymbol{x}, \boldsymbol{y}) < v) + u \cdot \mathbb{P}(V(\boldsymbol{x}, \boldsymbol{y}) = v)$ for any $v \in \mathbb{R}$ and $u \in [0,1]$. 
    Assuming the choice of $\boldsymbol{r}_{n+j} \equiv \boldsymbol{r} \in R$ is fixed,
    define 
    \begin{align}
    \label{eq:t_star}
        t^* = \sup \Big\{t \in [0,1]: \frac{t}{\mathbb{P}(F(V(\boldsymbol{x}, \boldsymbol{r}), U) \leq t)} \leq q\Big\} .
    \end{align}
    Suppose that, for any sufficiently small $\epsilon > 0$, there exists some $t \in (t^* - \epsilon, t^*)$ such that $\frac{t}{\mathbb{P}(F(V(\boldsymbol{x}, \boldsymbol{r}), U) \leq t)} \leq q$. Then the output $\mathcal{S}$ of Algorithm~\ref{algo:mcs} from input $\{(\boldsymbol{x}_i, \boldsymbol{y}_i)\}_{i=1}^{n} \cup \{\boldsymbol{x}_{n+j}\}_{j=1}^m$ satisfies
    \begin{align}
        &\lim_{n, m\rightarrow \infty} \textnormal{FDR} = \frac{\mathbb{P}(F(V(\boldsymbol{x}, \boldsymbol{r}), U) \leq t^*, \boldsymbol{y} \in R^c)}{\mathbb{P}(F(V(\boldsymbol{x}, \boldsymbol{r}), U) \leq t^*)}, \quad \textnormal{and } \nonumber \\
        &\lim_{n, m\rightarrow \infty} \textnormal{Power} = \frac{\mathbb{P}(F(V(\boldsymbol{x}, \boldsymbol{r}), U) \leq t^*, \boldsymbol{y} \in R)}{\mathbb{P}(\boldsymbol{y} \in R)}. \label{eq:asymp_power}
    \end{align}
\end{theorem}

We omit the proof of Theorem~\ref{thm:heuristic}, as it closely mirrors Proposition~7 in \citet{jin2023selection}. An intuitive explanation can be found in Appendix~\ref{sec:adv_clip}. Leveraging a characterization of the BH procedure~\citep{storey2004strong}, the theorem establishes that the asymptotics of FDR and power can be precisely achieved by replacing each term in \eqref{eq:fdr} and \eqref{eq:power} with its population counterpart. Notably, in this context, $t^*$ represents the population version of the BH rejection threshold for the $p$-values.

Theorem~\ref{thm:heuristic} indicates that the second, clipped score \eqref{eq:dist_score_2} is preferable to the first score \eqref{eq:dist_score_1} for achieving higher power. 
According to \eqref{eq:asymp_power}, since the value $V(\boldsymbol{x}, \boldsymbol{r})=-\inf_{\boldsymbol{s} \in R^c} \|\boldsymbol{r} - \boldsymbol{s}\|_p$ is identical for both scores assuming $\boldsymbol{r} \in \partial R$, it suffices to compare their asymptotic BH thresholds $t^*$. The score with the larger $t^*$ achieves higher asymptotic power and is therefore  more effective.
In the definition of $t^*$ in \eqref{eq:t_star}, the fraction can be rewritten as
\begin{align*}
    G_V(t) \equiv \frac{t}{\mathbb{P}(F(V(\boldsymbol{x}, \boldsymbol{r}), U) \leq t)} 
    &= \frac{\mathbb{P}(F(V(\boldsymbol{x}, \boldsymbol{y}), U) \leq t)}{\mathbb{P}(F(V(\boldsymbol{x}, \boldsymbol{r}), U) \leq t)} \\
    &= \frac{\mathbb{P}(V(\boldsymbol{x}, \boldsymbol{y}) \leq s)}{\mathbb{P}(V(\boldsymbol{x}, \boldsymbol{r}) \leq s)}
\end{align*}
when $s$ is the inverse of $F(\cdot, U)$ at $t$, assuming it exists. The first equality follows from the fact that $F(V(\boldsymbol{x}, \boldsymbol{y}), U) \sim \text{Unif}(0, 1)$, while the second arises from the monotonicity of $F$ with respect to its first argument $v$.
To maximize $t^*$ in \eqref{eq:t_star} (equivalently, the inverse $s^*$), an effective score should yield a larger $V(\boldsymbol{x}, \boldsymbol{y})$ relative to $V(\boldsymbol{x}, \boldsymbol{r})$, thereby reducing $G_{V}(t)$ for a fixed $t$. This, in turn, results in a larger $t^*$ when computing $\sup\{t:G_{V}(t)\leq q\}$ in \eqref{eq:t_star}. 

This criterion is precisely satisfied by the clipped score. 
An alternative justification for favoring the clipped score, based on maximizing the realized FDR, is provided in Appendix~\ref{sec:adv_clip},  along with further discussions on Theorem~\ref{thm:heuristic}.
For an empirical comparison of the performance of the two scores \eqref{eq:dist_score_1} and \eqref{eq:dist_score_2}, refer to Appendix~\ref{extra_experiments:score_comp}.



\subsection{\texttt{mCS-learn}: Learning-based Nonconformity Scores}

The two distance-based nonconformity scores introduced in the previous section offer straightforward and practical solutions for many scenarios. However, their effectiveness, particularly in the design of the second distance term $D_2(\cdot)$, is limited in some cases. For example, our numerical simulations indicate that \texttt{mCS-dist} would only achieve suboptimal power when $R$ is a nonconvex set; see Appendix~\ref{nonconvex} for further details. Furthermore, when the target region $R$ is irregular, constructing a closed-form distance function can be challenging,  leading to higher computational costs and potential inaccuracies.

To address these challenges, we propose an alternative method \texttt{mCS-learn}, which leverages a loss function that penalizes either the smooth selection size or conformal $p$-value to learn an optimal nonconformity score within the following family:
\begin{align} \label{eq:score_class}
    V^\theta(\boldsymbol{x}, \boldsymbol{y}) = M \cdot\mathbbm{1}\{ \boldsymbol{y} \notin R^c \cup \partial R\} - f_\theta(\boldsymbol{x}, \boldsymbol{y}; R)
\end{align}
where $M$ is a large constant and $f_\theta: \mathcal{X} \times \mathcal{Y} \rightarrow \mathbb{R}$ is a flexible function  parametrized by $\theta$, that can be chosen from a specific machine learning model class, such as kernel machines, gradient boosting models or neural networks, etc. 
The first term, an indicator function identical to $D_1(\cdot)$ in~\eqref{eq:dist_score_2}, ensures regional monotonicity in Definition~\ref{def:monotone} and boosts selection power, as suggested in Section~\ref{sec:dist}.
The second term generalizes the distance term $D_2(\cdot)$ from \eqref{eq:dist_score}, offering a more expressive framework for constructing optimal nonconformity scores.

The following result demonstrates the expressiveness of the family~\eqref{eq:score_class} by showing that it can include the optimal nonconformity score for any  selection task. The proof is provided in Appendix~\ref{proof:expressive}.

\begin{proposition} \label{prop:expressive}
    Let $\{(\boldsymbol{x}_i, \boldsymbol{y}_i)\}_{i=1}^{n+m}$ be sampled \textit{i.i.d.} from a distribution, and assume a fixed choice of $\boldsymbol{r}_{n+j} \equiv \boldsymbol{r}$.
    Under Algorithm~\ref{algo:mcs}, for any nominal FDR level $q$ and target region $R$, there exists a function $f^*$ such that the score constructed using $f^*$ in~\eqref{eq:score_class} maximizes the number of selected samples (and thus the power) among all scores with FDR control.
\end{proposition}

\begin{remark}[Subfamilies of the score class]
    A notable subfamily of \eqref{eq:score_class} is 
    \begin{align} \label{eq:family}
        M \cdot\mathbbm{1}\{ \boldsymbol{y} \notin R^c \cup \partial R\} - f_\theta(\hat\mu(\boldsymbol{x}); R),
    \end{align}
    where $f_{\theta}$ depends solely on the prediction $\hat{\mu}(\boldsymbol{x})$. This subfamily includes the clipped distance-based score from  Section \ref{sec:dist} as a special case. Here, regional monotonicity is guaranteed for any constant $M$, but the score can also be generalized further by allowing $f$ to depend on $\boldsymbol{x}$ as well:
    \begin{align*}
        M \cdot\mathbbm{1}\{ \boldsymbol{y} \notin R^c \cup \partial R\} - f_\theta(\boldsymbol{x}, \hat\mu(\boldsymbol{x}); R).
    \end{align*}
    In contrast, the broader family defined in \eqref{eq:score_class} offers greater flexibility by incorporating $\boldsymbol{y}$ in the second term. However, this added expressiveness  requires a sufficiently large $M$ to preserve regional monotonicity. Specifically, $M$ must satisfy $M > 2 |f_\theta(\boldsymbol{x}, \boldsymbol{y}; R)|$.
    In practice, $M$ is chosen to be sufficiently large to ensure that this inequality holds across the entire dataset.
\end{remark}

\begin{remark}[Incorporating pretrained models]
    
    This could be achieved in several ways, such as using the predictions of $\hat\mu$ as inputs to $f_\theta$, or train $f_\theta$ as a prediction on top of $\hat\mu$ when both models are implemented as neural networks. 
    While such practice does not increase the expressiveness of the score family, it often facilitates the training of $f_\theta$, as $\hat\mu(\boldsymbol{x})$ estimates $\boldsymbol{y}$ and is very informative for selection. Since $f_\theta$ can directly learn the data and the selection task, \texttt{mCS-learn} can still perform well when $\hat\mu$ is poorly fitted; see Appendix~\ref{other_mu}.
\end{remark}

To identify an optimal function within the family~\eqref{eq:score_class}, we introduce a differentiable loss function that mimics the inherently non-differentiable mCS procedure. The ``hard" sorting and ranking operations in the mCS workflow are replaced with their smooth, differentiable counterparts~\citep{blondel2020fast, cuturi2019differentiable}. 
We adopt the implementation introduced in~\citet{blondel2020fast}, with $\ell_2$ regularization and regularization strength set to 0.1.
The resulting loss function is then used for a chosen machine learning method to train $f_{\theta}$. Specifically, we partition the calibration data into three batches
$\mathcal{D}_{\textnormal{cal}} = \mathcal{D}_{f\textnormal{-train}} \cup \mathcal{D}_{f\textnormal{-val}} \cup \mathcal{D}'_{\textnormal{cal}}$, where $\mathcal{D}_{f\textnormal{-train}}$ and $\mathcal{D}_{f\textnormal{-val}}$ are used for training and validating $f_\theta$, respectively. Upon completion of training and validation, Algorithm~\ref{algo:mcs} can be applied with  $\mathcal{D}'_{\textnormal{cal}}$ as the calibration dataset for $V^{\theta}$ to generate the final selection set $\mathcal{S}$.

\paragraph{Training Step.}
The training loss function is defined based on the smoothed selection size. We denote the softened rank of an element $a \in A$ within the set $A$ by $\textnormal{soft-rank}(a; A)$. We randomly partition $\mathcal{D}_{f\textnormal{-train}}$ into two subsets $\mathcal{D}_{f\textnormal{-train1}}$ and $\mathcal{D}_{f\textnormal{-train2}}$, and we assume $\mathcal{D}_{f\textnormal{-train1}} = \{(\boldsymbol{x}_i, \boldsymbol{y}_i)\}_{i=1}^{n'}$ and $\mathcal{D}_{f\textnormal{-train2}} = \{(\boldsymbol{x}_{n'+j}, \boldsymbol{y}_{n'+j})\}_{j=1}^{m'}$ for notational simplicity. We define the smooth conformal $p$-value $\bar{p}_j^\theta$ for $j=1, \dots, m'$ as
\begin{align} \label{eq:soft_p}
    \bar{p}_j^\theta = \frac{\textnormal{soft-rank}\big(\widehat{V}_{n'+j}^\theta; \{V_i^\theta\}_{i =1}^{n'} \cup \{\widehat{V}_{n'+j}^\theta\}\big)}{n'+1}.
\end{align}
where $V_i^\theta := V^\theta(\boldsymbol{x}_{i}, \boldsymbol{y}_i)$ are computed on $\mathcal{D}_{f\textnormal{-train1}}$ and $\widehat{V}_{n'+j}^\theta := V^\theta(\boldsymbol{x}_{n'+j}, \boldsymbol{r}_{n'+j})$ are computed on $\mathcal{D}_{f\textnormal{-train2}}$, with $V^\theta$ in~\eqref{eq:score_class}. Here, $\mathcal{D}_{f\textnormal{-train1}}$ and $\mathcal{D}_{f\textnormal{-train2}}$ serve as the calibration dataset and test dataset respectively, to obtain the smoothened conformal $p$-value for training.

Next, to smooth the BH procedure, we first apply the soft-sorting operation to the smooth $p$-values  $\bar{p}_j^\theta$ to obtain their corresponding ranks $a_j^\theta$, 
\begin{align} \label{eq:rank_soft_p}
    a_j^\theta = \textnormal{soft-rank}\big(\bar{p}^\theta_j; \{\bar{p}^\theta_k\}_{k=1}^{m'}),
\end{align}
and then compute the softened selection size $\overline{S}(\theta)$ as
\begin{align*}
    \overline{S}(\theta) \overset{(i)}{=} \log \sum_{j=1}^{m'} e^{a_j^\theta s_j^\theta}, \  \text{where} \ 
    s_j^\theta \overset{(ii)}{=} \sigma \Big(\frac{q \cdot a_j^\theta / m' - \bar{p}^\theta_j}{\tau} \Big).
\end{align*}
In the above equation, the sigmoid function $\sigma$ with temperature coefficient $\tau$ in $(ii)$ serves as a smooth approximation of the indicator function $\mathbbm{1}\{\bar{p}_j^\theta < qa_j^\theta /m'\}$, while the log-sum-exp function in $(i)$ approximates the element-wise max function $\textnormal{max}(a_1^\theta s_1^\theta, \dots, a_{m'}^\theta s_{m'}^\theta)$. This formulation ensures that $\overline{S}(\theta)$ closely approximates the BH selection size, as the selection size is defined by the largest rank with a $p$-value below the threshold.

To maximize the selection size, the loss function for learning the \texttt{mCS-learn} score can be defined as
\begin{align} \label{eq:loss}
    L_1(\theta) = -\overline{S}(\theta).
\end{align}
While above formulation of the loss $L_1$ is intuitive in its attempt to approximate the final selection size using differentiable functions, the inclusion of two soft sorting steps may reduce numerical stability and impede the training process.

In the mCS procedure, the BH procedure is solely intended for multiplicity correction for FDR control, and thus is not necessarily required in the formulation of the loss function, whose primary objective is to learn the function $f_{\theta}$.  A simpler yet effective alternative loss function directly penalizes the smooth $p$-values $\bar{p}_j^\theta$ through the following loss function:
\begin{equation} \label{eq:alt_loss}
    L_2(\theta) = \sum_{j=1}^{m'} \bar{p}_j^\theta \big[ \mathbbm{1}(\boldsymbol{y}_{n+j} \in R) - \gamma \cdot \mathbbm{1}(\boldsymbol{y}_{n+j} \in R^c) \big].
\end{equation}
where $\bar{p}^\theta_j$ is minimized when the $j$-th sample is deemed desirable, as indicated by the first term. The second term, scaled by a balancing coefficient $\gamma$, ensures that the $p$-value is not uniformly small but becomes relatively larger for less favorable samples. 
This approach eliminates the need to estimate the BH threshold via a secondary soft-ordering~\eqref{eq:rank_soft_p}, leading to improved numerical stability and enhanced overall performance.
A comparison of the two loss functions can be found in Appendix~\ref{extra_experiments:loss_comp}. 
After computing the loss in each epoch, we can follow the standard backpropagation procedure to train $\theta$.

\paragraph{Validation Step.} To avoid overfitting  $f_{\theta}$, we perform an additional model selection procedure using a hold-out dataset $\mathcal{D}_{f\textnormal{-val}}$. Specifically, for each epoch $t=1, \dots, T$ of the backpropagation procedure, we apply $K$ random partitions on $\mathcal{D}_{f\textnormal{-val}}$ to obtain $\mathcal{D}_{f\textnormal{-val1}}^{(k)}$ and  $\mathcal{D}_{f\textnormal{-val2}}^{(k)}$ for $k = 1, \dots, K$. For each $k$, we then apply Algorithm~\ref{algo:mcs} with the setting $\mathcal{D}_{\textnormal{cal}}\coloneqq\mathcal{D}_{f\textnormal{-val1}}^{(k)}$ and $\mathcal{D}_{\textnormal{test}}\coloneqq\mathcal{D}_{f\textnormal{-val2}}^{(k)}$, and record the validation power $\rho_k(t)$. We then compute the average validation power $\overline{\text{Power}}(t)=\sum_{k=1}^{K}\rho_{k}(t)/K$.

In the end, we select $\bar{t}$ from $\{1, \dots, T\}$ to be the epoch with the highest average validation selection power $\overline{\text{Power}}(t)$, and deploy the associated model $f_{\theta_{\bar{t}}}$ for the final selection.

Finally, Algorithm~\ref{algo:training} details the complete learning procedure for \texttt{mCS-learn} scores.


\begin{algorithm}[h]
\caption{\texttt{mCS-learn} Learning Procedure}
\label{algo:training}
\begin{algorithmic}[1]
    
    \REQUIRE{
    Training data $\mathcal{D}_{f\textnormal{-train}}$, validation data $\mathcal{D}_{f\textnormal{-val}}$, target region $R$, FDR level $q \in (0,1)$, other hyperparameters.}
    \ENSURE{Trained nonconformity function $f_\theta$.}\
    \STATE Initialize parameters $\theta = \theta_0$.
    \FOR{epoch $t=1, \dots, T$}
        \begin{tcolorbox}[myalgoboxSkyBlue, title={Training Step}]
        \STATE Randomly partition $\mathcal{D}_{f\textnormal{-train}}$ into two disjoint subsets $\mathcal{D}_{f\textnormal{-train1}}$ and $\mathcal{D}_{f\textnormal{-train2}}$.
        \STATE Use the current $f_{\theta}$ to obtain $V_i^\theta$ from $\mathcal{D}_{f\textnormal{-train1}}$ and $\widehat{V}^\theta_{n'+j}$ from $\mathcal{D}_{f\textnormal{-train2}}$.
        \STATE Compute the smooth conformal $p$-values $\bar{p}_j^\theta$~\eqref{eq:soft_p}.
        \STATE Compute the loss function using~\eqref{eq:loss} or~\eqref{eq:alt_loss}.
        \STATE Back-propagate to update the parameters $\theta=\theta_t$.
        \end{tcolorbox}
        \begin{tcolorbox}[myalgoboxLavender, title={Validation Step}]
        \STATE Apply $K$ random partitions on $\mathcal{D}_{f\textnormal{-val}}$ to obtain $\mathcal{D}_{f\textnormal{-val1}}^{(k)}$ and  $\mathcal{D}_{f\textnormal{-val2}}^{(k)}$ for each $k=1, \dots, K$.
        \STATE 
        For each $k$, apply Algorithm~\ref{algo:mcs} for score function $V^{\theta}$ with the setting $\mathcal{D}_{\textnormal{cal}} \coloneqq\mathcal{D}^{(k)}_{f\textnormal{-val1}}$, $\mathcal{D}_{\textnormal{test}}\coloneqq \mathcal{D}^{(k)}_{f\textnormal{-val2}}$ and compute validation power $\rho_{k}(t)$.
        \STATE Compute $\overline{\text{Power}}(t)=\sum_{k=1}^{K}\rho_{k}(t)/K$.
        \end{tcolorbox}
    \ENDFOR
    \STATE Determine $\bar{t}=\arg\max_{t}\overline{\text{Power}}(t)$ and return $f_{\theta_{\bar{t}}}$.
\end{algorithmic}
\end{algorithm}

\section{Simulation Studies} \label{sec:simu}
\subsection{Baseline Methods} \label{sec:baseline}

While the standard CS approach is originally designed for univariate settings and cannot be directly applied to multivariate selections, appropriate adaptations can be devised. In this section, we introduce several na\"ive methods directly adapted from CS to address the multivariate case. Later, we employ these adapted methods as baselines and compare their performance against our proposed method.


In the scenario when the target region $R \subseteq \mathbb{R}^d$ is rectangular, the overall selection criterion can be decomposed, allowing each dimension to be evaluated independently.
By applying CS separately to each dimension, we obtain $d$ selection sets $\mathcal{S}_1, \dots, \mathcal{S}_d$, where each $\mathcal{S}_k$ contains observations satisfying the $k$-th corresponding marginal criterion. The final selection set $S$ is then given by the intersection of these individual sets, $\mathcal{S} = \cap_{k=1}^d \mathcal{S}_k$. We refer to this approach as \texttt{CS\_int}.

It can be shown that \texttt{CS\_int}, when each CS subroutine is conducted at a nominal level $q$, fails to control the FDR at or below $q$. This issue is analogous to the intersection hypothesis testing (IHT) problem in statistics. A common approach to address this is the Bonferroni correction \cite{dunn1961multiple}, which adjusts the nominal level of each subroutine to a lower threshold $q/d$. Then obtain the intersection of individual sets. However, this method is widely recognized as being overly conservative \cite{perneger1998s, westfall1993resampling}. We refer to the Bonferroni-adjusted \texttt{CS\_int} as \texttt{CS\_ib}. 
An alternative to the Bonferroni correction is to adaptively account for intersection hypothesis testing. Rather than predefining the nominal levels for subroutines, we determine suitable values by validating on a hold-out dataset. This approach necessitates additional data splits to construct the hold-out set, and we refer to it as \texttt{CS\_is}.

Beyond considering each dimension separately, another natural adaptation of CS involves transforming the response vector $\boldsymbol{y}$ before applying CS. Specifically, each response $\boldsymbol{y}_i$ is converted to a binary indicator reflecting whether it meets the selection criterion, defined as $\Tilde{y}_i = \mathbbm{1}\{\boldsymbol{y}_i \in R\}$. Under this transformation, the new selection threshold can be set to $c = 0$, as $\Tilde{y}_i > 0$ is equivalent to $\boldsymbol{y}_i \in R$. We refer to this approach as \texttt{bi}.    

\subsection{Numerical Results} \label{subsec:simu}

We compare the performance of \texttt{mCS-dist}, \texttt{mCS-learn} (abbreviated as \texttt{mCS-d} and \texttt{mCS-l} respectively) against the baseline methods outlined in Section~\ref{sec:baseline}.    

In our data generation processes, covariates $\boldsymbol{x}$ are sampled uniformly from $\text{Unif}(-1, 1)^p$ where $p$ is the covariate dimension, and the responses $\boldsymbol{y}$ are generated as $\boldsymbol{y} = \mu(\boldsymbol{x}) + \boldsymbol{\epsilon}$, where $\mu$ denotes the regression function and $\boldsymbol{\epsilon}$ represents noise drawn from either a multivariate Gaussian or multivariate $t$-distribution. By varying the regression function, the size of response dimensions, and the choice of Gaussian or heavy-tailed noise, we create a range of selection problems with differing levels of difficulty.

We consider two selection tasks where the target region $R$ is defined as:
\begin{enumerate}[itemsep=-1pt, topsep=1pt, leftmargin=35pt]
    \item[\textbf{Task 1}.] The (shifted) first orthant, $R = \{\boldsymbol{y}: y_k \geq c_k \hspace{0.2cm} \forall k\}$,    \item[\textbf{Task 2}.] A sphere centered at $\boldsymbol{c}$, $R = \{\boldsymbol{y}: \|\boldsymbol{y} - \boldsymbol{c}\|_2 \leq r\}$.
\end{enumerate}
These two specific tasks are particularly relevant in  applications, as they simulate scenarios where (1) $d$ criteria must be simultaneously satisfied or (2) an instance must be sufficiently close to a specific point $\boldsymbol{c}$ to be deemed acceptable.

In our simulation, the coefficients $\boldsymbol{c}$ and $r$ for each selection problem are chosen to ensure that approximately 15\% to 35\% of the data points lie within $R$ across all six data-generating processes. Detailed descriptions of the data generation process, model specifications, and the specific values of the coefficients are provided in Appendix~\ref{sec:data_gen}.
We first train a support vector regression model $\hat\mu$ using 1000 data points, and use an additional labeled dataset of 1000 samples to construct selection sets for different methods in comparison. We evaluate the selection power and FDR using a test dataset of size 100. 
We adopt the clipped score \eqref{eq:dist_score_2} for \texttt{mCS-dist}, and adopt the loss function in~\eqref{eq:alt_loss} with balancing coefficient $\gamma = 0.5$ for \texttt{mCS-learn}. 
For \texttt{mCS-learn}, the calibration data is split to $\mathcal{D}_{f\textnormal{-train}}, \mathcal{D}_{f\textnormal{-val}}$ and $\mathcal{D}'_{\textnormal{cal}}$ with ratio 8:1:1, and the model $f_\theta$
is formulated as a two-layer MLP with batch normalization.
The response dimension is set to be $d=30$, and nominal FDR level is set at $q=0.3$. Number of iterations for validation is set to $K=100$. The selection process is repeated across 100 iterations, with a new dataset generated independently for each iteration.

Table \ref{tab:fdp1} and Table \ref{tab:power1} summarize the experimental results for the first selection task. 
As shown in Table \ref{tab:fdp1}, \texttt{CS\_int} substantially violates FDR control. 
\texttt{CS\_is} provides only approximate FDR control, and in scenarios such as Setting 6, the FDR control may be compromised.
In each setting, the red numbers in Table \ref{tab:power1} indicate the highest achieved power among all methods that properly control the FDR. 
Among the four remaining methods that always maintain valid FDR control, our two proposed methods consistently achieve the best and second-best power, outperforming the baseline methods under all settings.

Table \ref{tab:power_and_fdp2} summarizes similar results for the second selection task. Baseline methods \texttt{CS\_int} and \texttt{CS\_ib} are not included as they are not applicable to non-rectangular target regions. Figure \ref{fig:varying_q} shows the realized FDR and power curves across varying nominal FDR levels, ranging from 0.05 to 0.5 in increments of 0.05. Results are shown exclusively for Setting 3 due to space constraints. Among the methods compared, \texttt{mCS-dist}, \texttt{mCS-learn}, and \texttt{bi} demonstrate consistent FDR control, as their respective curves remain below the black dashed line indicating the nominal FDR threshold. Notably, \texttt{mCS-learn} also achieves consistently higher power across all nominal levels.

Additional simulation results, including ablation studies exploring various factors of the selection problem and our algorithm, are provided in Appendix~\ref{extra_experiments}.

\setlength{\tabcolsep}{2pt} 
\begin{table}[htbp]
\centering
\small
\caption{Observed FDR for Task 1 (shifted first orthant $R$).  The nominal FDR level is $q = 0.3$.}

\label{tab:fdp1}
\begin{adjustbox}{max width=0.7\textwidth}
\begin{tabular}{lcccccc}
\toprule

Setting & \texttt{CS\_int} & \texttt{CS\_ib} & \texttt{CS\_is} & \texttt{bi} & \texttt{mCS-d} & \texttt{mCS-l}\\
\midrule
1 & $0.773$ & $0.000$ & $0.280$ & $0.240$ & $0.277$ & $0.251$ \\
2 & $0.801$ & $0.000$ & $0.133$ & $0.300$ & $0.315$ & $0.266$ \\
3 & $0.724$ & $0.000$ & $0.204$ & $0.295$ & $0.264$ & $0.278$ \\
4 & $0.811$ & $0.000$ & $0.255$ & $0.309$ & $0.277$ & $0.315$ \\
5 & $0.810$ & $0.000$ & $0.300$ & $0.266$ & $0.308$ & $0.239$ \\
6 & $0.778$ & $0.000$ & $0.374$ & $0.245$ & $0.287$ & $0.258$ \\
\bottomrule
\end{tabular}
\end{adjustbox}
\end{table}

\vspace*{-0.7em}

\begin{table}[htbp]
\centering
\small
\caption{Observed power for Task 1 (shifted first orthant $R$).}

\label{tab:power1}
\begin{adjustbox}{max width=0.7\textwidth}
\begin{tabular}{lcccccc}
\toprule

Setting & \texttt{CS\_int} & \texttt{CS\_ib} & \texttt{CS\_is} & \texttt{bi} & \texttt{mCS-d} & \texttt{mCS-l}\\
\midrule
1 & $1.000$ & $0.000$ & $0.406$ & $0.324$ & $\color{red}0.555$ & $0.325$ \\
2 & $1.000$ & $0.000$ & $0.012$ & $0.069$ & $0.104$ & $\color{red}0.108$ \\
3 & $1.000$ & $0.000$ & $0.039$ & $0.059$ & $0.068$ & $\color{red}0.102$ \\
4 & $1.000$ & $0.000$ & $0.124$ & $0.194$ & $\color{red}0.324$ & $0.198$ \\
5 & $1.000$ & $0.000$ & $0.019$ & $0.035$ & $\color{red}0.060$ & $0.042$ \\
6 & $1.000$ & $0.000$ & $0.101$ & $0.027$ & $\color{red}0.046$ & $0.034$ \\
\bottomrule
\end{tabular}
\end{adjustbox}
\end{table}

%


%


\vspace*{-0.7em}

\begin{table}[htbp]
\centering
\small
\setlength{\tabcolsep}{4pt}
\caption{Observed FDR and power for Task 2 (spherical $R$). The nominal FDR level is $q = 0.3$.}
\label{tab:power_and_fdp2}

\begin{tabular}{lcccccc}
\toprule
& \multicolumn{3}{c}{FDR} & \multicolumn{3}{c}{Power} \\
\cmidrule(lr){2-4} \cmidrule(lr){5-7}
Setting & \texttt{bi} & \texttt{mCS-d} & \texttt{mCS-l} & \texttt{bi} & \texttt{mCS-d} & \texttt{mCS-l} \\
\midrule
1 & $0.255$ & $0.265$ & $0.279$ & $0.636$ & $\color{red}0.760$ & $0.534$ \\
2 & $0.260$ & $0.263$ & $0.273$ & $0.343$ & $0.405$ & $\color{red}0.421$ \\
3 & $0.216$ & $0.223$ & $0.263$ & $0.115$ & $0.134$ & $\color{red}0.179$ \\
4 & $0.316$ & $0.286$ & $0.254$ & $0.192$ & $\color{red}0.333$ & $0.180$ \\
5 & $0.307$ & $0.291$ & $0.273$ & $0.137$ & $0.170$ & $\color{red}0.189$ \\
6 & $0.292$ & $0.283$ & $0.207$ & $0.055$ & $\color{red}0.063$ & $0.061$ \\
\bottomrule
\end{tabular}
\end{table}

\begin{figure*}[!h]
    \centering
    \includegraphics[width=1\textwidth]{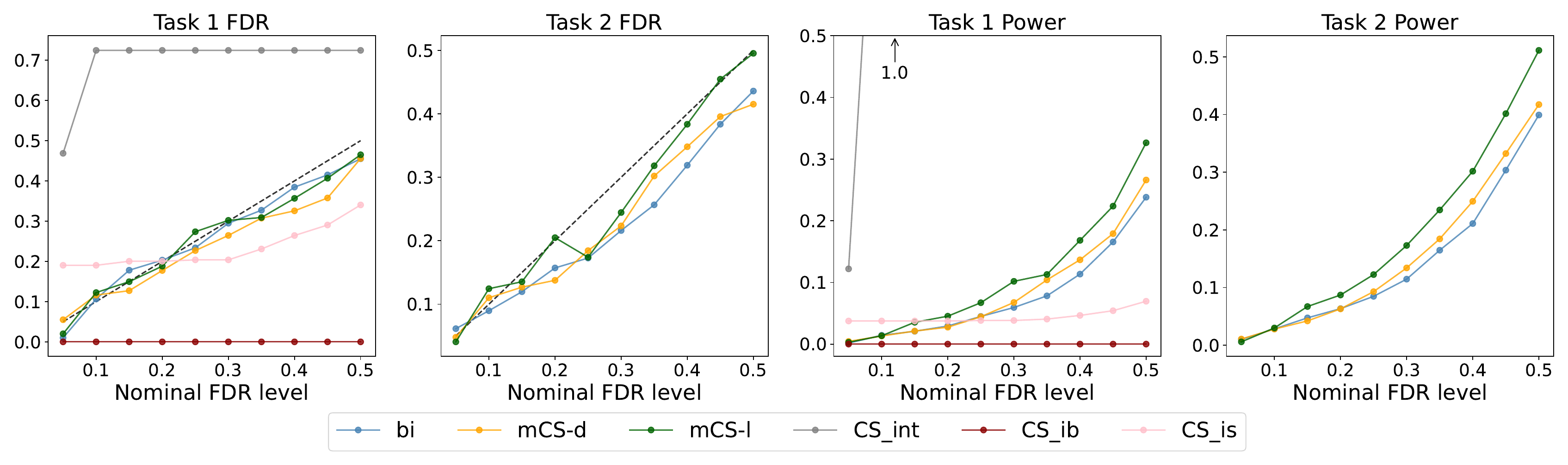}
    \caption{Observed FDR and power across varying nominal levels for Task 1 (shifted first orthant $R$) and 2 (spherical $R$).}
    \label{fig:varying_q}
\end{figure*}


\section{Real Data Application} \label{sec:real_data}

We apply the mCS framework to drug discovery, selecting drug candidates with desirable biological properties. Each candidate corresponds to a chemical compound, where the feature vector $\boldsymbol{x}$ encodes structural or chemical characteristics, and the multivariate response $\boldsymbol{y}$ represents biological properties. The multidimensional nature of $\boldsymbol{y}$ reflects the need to evaluate multiple biological criteria before advancing a compound. Ensuring FDR control improves downstream processes, such as wet-lab validation.

We employ an imputed public ADMET dataset compiled from multiple sources \cite{wenzel2019predictive, iwata2022predicting, kim2023pubchem, watanabe2018predicting, falcon2022reliable, esposito2020combining, braga2015pred, aliagas2022comparison, perryman2020pruned, meng2022boosting, vermeire2022predicting}, comprising
$n = 22805$ compounds with $d=15$ biological assay responses. We focus on three different selection tasks with the following target regions: 
\begin{enumerate}[itemsep=-1pt, topsep=1pt, leftmargin=35pt]
    \item[\textbf{Task 1}.] The (shifted) first orthant, $R = \{\boldsymbol{y}: y_k \geq c_k \hspace{0.2cm} \forall k\}$,
    \item[\textbf{Task 2}.] A sphere centered at $\boldsymbol{c}$, $R = \{\boldsymbol{y}: \|\boldsymbol{y} - \boldsymbol{c}\|_2 \leq r\}$,
    \item[\textbf{Task 3}.] The complement of a sphere centered at $\boldsymbol{c}$, $R = \{ \boldsymbol{y}: \|\boldsymbol{y} - \boldsymbol{c}\|_2 \geq r'\}$.
\end{enumerate}
We included the two tasks introduced in the numerical simulations (Section~\ref{sec:simu}), and we also designed the third task to evaluate the performance of various methods under a nonconvex target region with real data.
Each of the selection task is designed so that approximately 15\%–30\% of the compounds qualify as acceptable. 
Further details on the dataset and problem setup, including response descriptions and cutoffs, can be found in Appendix~\ref{real_results}. 

In this experiment, the underlying predictor $\hat\mu$ is specified as a drug property prediction model from the \texttt{DeepPurpose} Python package~\cite{huang2020deeppurpose} with \texttt{Morgan} drug encoding. We train the model using $n_\textnormal{train} = 12000$ samples, provide $n_\textnormal{cal} = 8000$ samples for calibration and reserves the remaining data of size $n_\textnormal{test} = 2805$ as test data. 
We keep the configuration and hyperparameters of the methods unchanged as in Section~\ref{sec:simu}.
Two nominal levels $q = 0.3$ and $q = 0.5$ are considered, and the selection processes for each method are repeated across 500 iterations.

\setlength{\tabcolsep}{2pt} 
\begin{table}[htbp]
\centering
\small
\caption{Observed FDR of different methods with real data.}

\label{tab:real_fdp}
\begin{adjustbox}{max width=0.7\textwidth}
\begin{tabular}{lccccccc}
\toprule

Task & $q$ & \texttt{CS\_int} & \texttt{CS\_ib} & \texttt{CS\_is} & \texttt{bi} & \texttt{mCS-d} & \texttt{mCS-l}\\
\midrule
1 & $0.3$ & $0.760$ & $0.000$ & $0.303$ & $0.038$ & $0.304$ & $0.275$ \\
2 & $0.3$ & $-$ & $-$ & $-$ & $0.000$ & $0.300$ & $0.293$ \\
3 & $0.3$ & $-$ & $-$ & $-$ & $0.084$ & $0.301$ & $0.296$ \\
\midrule
1 & $0.5$ & $0.782$ & $0.393$ & $0.496$ & $0.040$ & $0.499$ & $0.488$ \\
2 & $0.5$ & $-$ & $-$ & $-$ & $0.000$ & $0.499$ & $0.498$ \\
3 & $0.5$ & $-$ & $-$ & $-$ & $0.084$ & $0.501$ & $0.497$ \\
\bottomrule
\end{tabular}
\end{adjustbox}
\end{table}

\begin{table}[htbp]
\centering
\small
\caption{Observed power of methods with real data.}

\label{tab:real_power}
\begin{adjustbox}{max width=0.7\textwidth}
\begin{tabular}{lccccccc}
\toprule

Task & $q$ & \texttt{CS\_int} & \texttt{CS\_ib} & \texttt{CS\_is} & \texttt{bi} & \texttt{mCS-d} & \texttt{mCS-l}\\
\midrule
1 & $0.3$ & $0.993$ & $0.000$ & $0.019$ & $0.000$ & $0.006$ & ${\color{red}0.010}$ \\
2 & $0.3$ & $-$ & $-$ & $-$ & $0.000$ & ${\color{red}0.278}$ & $0.086$ \\
3 & $0.3$ & $-$ & $-$ & $-$ & $0.000$ & $0.410$ & ${\color{red}0.431}$ \\
\midrule
1 & $0.5$ & $1.000$ & $0.003$ & $0.225$ & $0.000$ & ${\color{red}0.433}$ & $0.193$ \\
2 & $0.5$ & $-$ & $-$ & $-$ & $0.000$ & ${\color{red}0.759}$ & $0.515$ \\
3 & $0.5$ & $-$ & $-$ & $-$ & $0.000$ & $0.449$ & ${\color{red}0.589}$ \\
\bottomrule
\end{tabular}
\end{adjustbox}
\end{table}

Tables~\ref{tab:real_fdp} and Table~\ref{tab:real_power} summarize the FDR and power of the competing methods respectively. 
Consistent with our numerical simulation, the methods \texttt{bi}, \texttt{mCS-dist}, and \texttt{mCS-learn} all demonstrate valid FDR control. 
Among the methods guaranteed to control FDR, \texttt{mCS-dist} and \texttt{mCS-learn} consistently achieve the best and the second-best power across all tasks and nominal levels.
Notably, \texttt{mCS-learn} exhibits superior performance under nonconvex target regions, corroborating the results presented in Appendix~\ref{nonconvex}.
We note that although the baseline method \texttt{bi} performed well in the simulation settings, it barely selected any compound in the current task. This outcome may be attributed to the suboptimal performance of the underlying binary classification model, which achieved an F1 score of only 0.31.

\section{Conclusion}

We propose multivariate conformal selection, an extension of conformal selection to multivariate response settings. Our experiments demonstrate that mCS significantly improves selection power while maintaining rigorous FDR control, outperforming existing baselines across simulated and real-world datasets. The flexibility of mCS makes it a valuable tool for selective tasks involved in diverse fields including drug discovery. Looking forward, we anticipate that the mCS framework can be further extended to handle additional practical challenges, including settings with hierarchical or structured responses. By addressing these challenges, mCS has the potential to further enhance its applicability and impact across diverse scientific and industrial domains.

\section{Impact Statement}

This paper aims to advance the field of machine learning. While there may be societal impacts, none require specific attention here.

\clearpage
\bibliography{ref}

\begin{thebibliography}{42}
\providecommand{\natexlab}[1]{#1}
\providecommand{\url}[1]{\texttt{#1}}
\expandafter\ifx\csname urlstyle\endcsname\relax
  \providecommand{\doi}[1]{doi: #1}\else
  \providecommand{\doi}{doi: \begingroup \urlstyle{rm}\Url}\fi

\bibitem[Aliagas et~al.(2022)Aliagas, Gobbi, Lee, and Sellers]{aliagas2022comparison}
Aliagas, I., Gobbi, A., Lee, M.-L., and Sellers, B.~D.
\newblock Comparison of logp and logd correction models trained with public and proprietary data sets.
\newblock \emph{Journal of Computer-Aided Molecular Design}, 36\penalty0 (3):\penalty0 253--262, 2022.

\bibitem[Bai et~al.(2022)Bai, Jones, Ndousse, Askell, Chen, DasSarma, Drain, Fort, Ganguli, Henighan, et~al.]{bai2022training}
Bai, Y., Jones, A., Ndousse, K., Askell, A., Chen, A., DasSarma, N., Drain, D., Fort, S., Ganguli, D., Henighan, T., et~al.
\newblock Training a helpful and harmless assistant with reinforcement learning from human feedback.
\newblock \emph{arXiv preprint arXiv:2204.05862}, 2022.

\bibitem[Bates et~al.(2021)Bates, Angelopoulos, Lei, Malik, and Jordan]{bates2021distribution}
Bates, S., Angelopoulos, A., Lei, L., Malik, J., and Jordan, M.
\newblock Distribution-free, risk-controlling prediction sets.
\newblock \emph{Journal of the ACM (JACM)}, 68\penalty0 (6):\penalty0 1--34, 2021.

\bibitem[Bates et~al.(2023)Bates, Cand\`es, Lei, Romano, and Sesia]{Bates2023}
Bates, S., Cand\`es, E., Lei, L., Romano, Y., and Sesia, M.
\newblock Testing for outliers with conformal p-values.
\newblock \emph{The Annals of Statistics}, 51\penalty0 (1), February 2023.
\newblock ISSN 0090-5364.
\newblock \doi{10.1214/22-aos2244}.
\newblock URL \url{http://dx.doi.org/10.1214/22-AOS2244}.

\bibitem[Benjamini \& Hochberg(1995)Benjamini and Hochberg]{bh1995}
Benjamini, Y. and Hochberg, Y.
\newblock Controlling the false discovery rate: A practical and powerful approach to multiple testing.
\newblock \emph{Journal of the Royal Statistical Society. Series B (Methodological)}, 57\penalty0 (1):\penalty0 289--300, 1995.
\newblock ISSN 00359246.
\newblock URL \url{http://www.jstor.org/stable/2346101}.

\bibitem[Benjamini \& Hochberg(1997)Benjamini and Hochberg]{bh1997}
Benjamini, Y. and Hochberg, Y.
\newblock Multiple hypotheses testing with weights.
\newblock \emph{Scandinavian Journal of Statistics}, 24\penalty0 (3):\penalty0 407--418, 1997.
\newblock \doi{https://doi.org/10.1111/1467-9469.00072}.
\newblock URL \url{https://onlinelibrary.wiley.com/doi/abs/10.1111/1467-9469.00072}.

\bibitem[Blondel et~al.(2020)Blondel, Teboul, Berthet, and Djolonga]{blondel2020fast}
Blondel, M., Teboul, O., Berthet, Q., and Djolonga, J.
\newblock Fast differentiable sorting and ranking.
\newblock In \emph{International Conference on Machine Learning}, pp.\  950--959. PMLR, 2020.

\bibitem[Braga et~al.(2015)Braga, Alves, Silva, Muratov, Fourches, Li{\~a}o, Tropsha, and Andrade]{braga2015pred}
Braga, R.~C., Alves, V.~M., Silva, M.~F., Muratov, E., Fourches, D., Li{\~a}o, L.~M., Tropsha, A., and Andrade, C.~H.
\newblock Pred-herg: A novel web-accessible computational tool for predicting cardiac toxicity.
\newblock \emph{Molecular Informatics}, 34\penalty0 (10):\penalty0 698--701, 2015.

\bibitem[Cuturi et~al.(2019)Cuturi, Teboul, and Vert]{cuturi2019differentiable}
Cuturi, M., Teboul, O., and Vert, J.-P.
\newblock Differentiable ranking and sorting using optimal transport.
\newblock \emph{Advances in Neural Information Processing Systems}, 32, 2019.

\bibitem[Dunn(1961)]{dunn1961multiple}
Dunn, O.~J.
\newblock Multiple comparisons among means.
\newblock \emph{Journal of the American Statistical Association}, 56\penalty0 (293):\penalty0 52--64, 1961.

\bibitem[Esposito et~al.(2020)Esposito, Wang, Lange, Oellien, and Riniker]{esposito2020combining}
Esposito, C., Wang, S., Lange, U.~E., Oellien, F., and Riniker, S.
\newblock Combining machine learning and molecular dynamics to predict p-glycoprotein substrates.
\newblock \emph{Journal of Chemical Information and Modeling}, 60\penalty0 (10):\penalty0 4730--4749, 2020.

\bibitem[Falc{\'o}n-Cano et~al.(2022)Falc{\'o}n-Cano, Molina, and Cabrera-P{\'e}rez]{falcon2022reliable}
Falc{\'o}n-Cano, G., Molina, C., and Cabrera-P{\'e}rez, M.~{\'A}.
\newblock Reliable prediction of caco-2 permeability by supervised recursive machine learning approaches.
\newblock \emph{Pharmaceutics}, 14\penalty0 (10):\penalty0 1998, 2022.

\bibitem[Feldman et~al.(2023)Feldman, Bates, and Romano]{feldman2023calibrated}
Feldman, S., Bates, S., and Romano, Y.
\newblock Calibrated multiple-output quantile regression with representation learning.
\newblock \emph{Journal of Machine Learning Research}, 24\penalty0 (24):\penalty0 1--48, 2023.

\bibitem[Gui et~al.(2024)Gui, Jin, and Ren]{gui2024conformal}
Gui, Y., Jin, Y., and Ren, Z.
\newblock Conformal alignment: Knowing when to trust foundation models with guarantees.
\newblock \emph{Advances in Neural Information Processing Systems}, 2024.

\bibitem[Heid et~al.(2023)Heid, Greenman, Chung, Li, Graff, Vermeire, Wu, Green, and McGill]{heid2023chemprop}
Heid, E., Greenman, K.~P., Chung, Y., Li, S.-C., Graff, D.~E., Vermeire, F.~H., Wu, H., Green, W.~H., and McGill, C.~J.
\newblock Chemprop: a machine learning package for chemical property prediction.
\newblock \emph{Journal of Chemical Information and Modeling}, 64\penalty0 (1):\penalty0 9--17, 2023.

\bibitem[Huang et~al.(2020)Huang, Fu, Glass, Zitnik, Xiao, and Sun]{huang2020deeppurpose}
Huang, K., Fu, T., Glass, L.~M., Zitnik, M., Xiao, C., and Sun, J.
\newblock Deeppurpose: a deep learning library for drug--target interaction prediction.
\newblock \emph{Bioinformatics}, 36\penalty0 (22-23):\penalty0 5545--5547, 2020.

\bibitem[Iwata et~al.(2022)Iwata, Matsuo, Mamada, Motomura, Matsushita, Fujiwara, Maeda, and Handa]{iwata2022predicting}
Iwata, H., Matsuo, T., Mamada, H., Motomura, T., Matsushita, M., Fujiwara, T., Maeda, K., and Handa, K.
\newblock Predicting total drug clearance and volumes of distribution using the machine learning-mediated multimodal method through the imputation of various nonclinical data.
\newblock \emph{Journal of Chemical Information and Modeling}, 62\penalty0 (17):\penalty0 4057--4065, 2022.

\bibitem[Jin \& Cand{\`e}s(2023)Jin and Cand{\`e}s]{jin2023selection}
Jin, Y. and Cand{\`e}s, E.~J.
\newblock Selection by prediction with conformal p-values.
\newblock \emph{Journal of Machine Learning Research}, 24\penalty0 (244):\penalty0 1--41, 2023.

\bibitem[Johnstone \& Cox(2021)Johnstone and Cox]{johnstone2021conformal}
Johnstone, C. and Cox, B.
\newblock Conformal uncertainty sets for robust optimization.
\newblock In \emph{Conformal and Probabilistic Prediction and Applications}, pp.\  72--90. PMLR, 2021.

\bibitem[Kim et~al.(2023)Kim, Chen, Cheng, Gindulyte, He, He, Li, Shoemaker, Thiessen, Yu, et~al.]{kim2023pubchem}
Kim, S., Chen, J., Cheng, T., Gindulyte, A., He, J., He, S., Li, Q., Shoemaker, B.~A., Thiessen, P.~A., Yu, B., et~al.
\newblock Pubchem 2023 update.
\newblock \emph{Nucleic Acids Research}, 51\penalty0 (D1):\penalty0 D1373--D1380, 2023.

\bibitem[Kivaranovic et~al.(2020)Kivaranovic, Johnson, and Leeb]{kivaranovic2020adaptive}
Kivaranovic, D., Johnson, K.~D., and Leeb, H.
\newblock Adaptive, distribution-free prediction intervals for deep networks.
\newblock In \emph{International Conference on Artificial Intelligence and Statistics}, pp.\  4346--4356. PMLR, 2020.

\bibitem[Klein et~al.(2025)Klein, Bethune, Ndiaye, and Cuturi]{klein2025multivariate}
Klein, M., Bethune, L., Ndiaye, E., and Cuturi, M.
\newblock Multivariate conformal prediction using optimal transport.
\newblock \emph{arXiv preprint arXiv:2502.03609}, 2025.

\bibitem[Kuleshov et~al.(2018)Kuleshov, Bernstein, and Burnaev]{kuleshov2018conformal}
Kuleshov, A., Bernstein, A., and Burnaev, E.
\newblock Conformal prediction in manifold learning.
\newblock In \emph{Conformal and Probabilistic Prediction and Applications}, pp.\  234--253. PMLR, 2018.

\bibitem[Lei \& Cand{\`e}s(2021)Lei and Cand{\`e}s]{lei2021conformal}
Lei, L. and Cand{\`e}s, E.~J.
\newblock Conformal inference of counterfactuals and individual treatment effects.
\newblock \emph{Journal of the Royal Statistical Society Series B: Statistical Methodology}, 83\penalty0 (5):\penalty0 911--938, 2021.

\bibitem[Meng et~al.(2022)Meng, Chen, Wahib, Yang, Zheng, Wei, Feng, and Liu]{meng2022boosting}
Meng, J., Chen, P., Wahib, M., Yang, M., Zheng, L., Wei, Y., Feng, S., and Liu, W.
\newblock Boosting the predictive performance with aqueous solubility dataset curation.
\newblock \emph{Scientific Data}, 9\penalty0 (1):\penalty0 71, 2022.

\bibitem[Messoudi et~al.(2022)Messoudi, Destercke, and Rousseau]{messoudi2022ellipsoidal}
Messoudi, S., Destercke, S., and Rousseau, S.
\newblock Ellipsoidal conformal inference for multi-target regression.
\newblock In \emph{Conformal and Probabilistic Prediction with Applications}, pp.\  294--306. PMLR, 2022.

\bibitem[Park et~al.(2024)Park, Tibshirani, and Cho]{park2024semiparametric}
Park, J.~W., Tibshirani, R., and Cho, K.
\newblock Semiparametric conformal prediction.
\newblock \emph{arXiv preprint arXiv:2411.02114}, 2024.

\bibitem[Perneger(1998)]{perneger1998s}
Perneger, T.~V.
\newblock What's wrong with bonferroni adjustments.
\newblock \emph{BMJ}, 316\penalty0 (7139):\penalty0 1236--1238, 1998.

\bibitem[Perryman et~al.(2020)Perryman, Inoyama, Patel, Ekins, and Freundlich]{perryman2020pruned}
Perryman, A.~L., Inoyama, D., Patel, J.~S., Ekins, S., and Freundlich, J.~S.
\newblock Pruned machine learning models to predict aqueous solubility.
\newblock \emph{ACS Omega}, 5\penalty0 (27):\penalty0 16562--16567, 2020.

\bibitem[Romano et~al.(2019)Romano, Patterson, and Candes]{romano2019conformalized}
Romano, Y., Patterson, E., and Candes, E.
\newblock Conformalized quantile regression.
\newblock \emph{Advances in Neural Information Processing Systems}, 32, 2019.

\bibitem[Scannell et~al.(2022)Scannell, Bosley, Hickman, Dawson, Truebel, Ferreira, Richards, and Treherne]{scannell2022predictive}
Scannell, J.~W., Bosley, J., Hickman, J.~A., Dawson, G.~R., Truebel, H., Ferreira, G.~S., Richards, D., and Treherne, J.~M.
\newblock Predictive validity in drug discovery: what it is, why it matters and how to improve it.
\newblock \emph{Nature Reviews Drug Discovery}, 21\penalty0 (12):\penalty0 915--931, 2022.

\bibitem[Sesia \& Cand{\`e}s(2020)Sesia and Cand{\`e}s]{sesia2020comparison}
Sesia, M. and Cand{\`e}s, E.~J.
\newblock A comparison of some conformal quantile regression methods.
\newblock \emph{Stat}, 9\penalty0 (1):\penalty0 e261, 2020.

\bibitem[Sheridan et~al.(2015)Sheridan, McMasters, Voigt, and Wildey]{sheridan2015ecounterscreening}
Sheridan, R.~P., McMasters, D.~R., Voigt, J.~H., and Wildey, M.~J.
\newblock ecounterscreening: using qsar predictions to prioritize testing for off-target activities and setting the balance between benefit and risk.
\newblock \emph{Journal of Chemical Information and Modeling}, 55\penalty0 (2):\penalty0 231--238, 2015.

\bibitem[Storey et~al.(2004)Storey, Taylor, and Siegmund]{storey2004strong}
Storey, J.~D., Taylor, J.~E., and Siegmund, D.
\newblock Strong control, conservative point estimation and simultaneous conservative consistency of false discovery rates: a unified approach.
\newblock \emph{Journal of the Royal Statistical Society Series B: Statistical Methodology}, 66\penalty0 (1):\penalty0 187--205, 2004.

\bibitem[Szyma{\'n}ski et~al.(2011)Szyma{\'n}ski, Markowicz, and Mikiciuk-Olasik]{szymanski2011adaptation}
Szyma{\'n}ski, P., Markowicz, M., and Mikiciuk-Olasik, E.
\newblock Adaptation of high-throughput screening in drug discovery—toxicological screening tests.
\newblock \emph{International Journal of Molecular Sciences}, 13\penalty0 (1):\penalty0 427--452, 2011.

\bibitem[Vermeire et~al.(2022)Vermeire, Chung, and Green]{vermeire2022predicting}
Vermeire, F.~H., Chung, Y., and Green, W.~H.
\newblock Predicting solubility limits of organic solutes for a wide range of solvents and temperatures.
\newblock \emph{Journal of the American Chemical Society}, 144\penalty0 (24):\penalty0 10785--10797, 2022.

\bibitem[Vovk et~al.(2005)Vovk, Gammerman, and Shafer]{Vovk2005}
Vovk, V., Gammerman, A., and Shafer, G.
\newblock Algorithmic learning in a random world.
\newblock 2005.
\newblock URL \url{https://api.semanticscholar.org/CorpusID:118783209}.

\bibitem[Watanabe et~al.(2018)Watanabe, Esaki, Kawashima, Natsume-Kitatani, Nagao, Ohashi, and Mizuguchi]{watanabe2018predicting}
Watanabe, R., Esaki, T., Kawashima, H., Natsume-Kitatani, Y., Nagao, C., Ohashi, R., and Mizuguchi, K.
\newblock Predicting fraction unbound in human plasma from chemical structure: improved accuracy in the low value ranges.
\newblock \emph{Molecular Pharmaceutics}, 15\penalty0 (11):\penalty0 5302--5311, 2018.

\bibitem[Wenzel et~al.(2019)Wenzel, Matter, and Schmidt]{wenzel2019predictive}
Wenzel, J., Matter, H., and Schmidt, F.
\newblock Predictive multitask deep neural network models for adme-tox properties: learning from large data sets.
\newblock \emph{Journal of Chemical Information and Modeling}, 59\penalty0 (3):\penalty0 1253--1268, 2019.

\bibitem[Westfall \& Young(1993)Westfall and Young]{westfall1993resampling}
Westfall, P.~H. and Young, S.~S.
\newblock \emph{Resampling-Based Multiple Testing: Examples and Methods for p-value Adjustment}, volume 279.
\newblock John Wiley \& Sons, 1993.

\bibitem[Yang et~al.(2019)Yang, Swanson, Jin, Coley, Eiden, Gao, Guzman-Perez, Hopper, Kelley, Mathea, et~al.]{yang2019analyzing}
Yang, K., Swanson, K., Jin, W., Coley, C., Eiden, P., Gao, H., Guzman-Perez, A., Hopper, T., Kelley, B., Mathea, M., et~al.
\newblock Analyzing learned molecular representations for property prediction.
\newblock \emph{Journal of Chemical Information and Modeling}, 59\penalty0 (8):\penalty0 3370--3388, 2019.

\bibitem[Zhang et~al.(2025)Zhang, Yang, Wang, Yu, Huang, Li, Li, Wu, and Yang]{zhang2025artificial}
Zhang, K., Yang, X., Wang, Y., Yu, Y., Huang, N., Li, G., Li, X., Wu, J.~C., and Yang, S.
\newblock Artificial intelligence in drug development.
\newblock \emph{Nature Medicine}, pp.\  1--15, 2025.

\end{thebibliography}
\bibliographystyle{icml2025}

\newpage
\appendix
\onecolumn
\section{Technical Proofs}

\subsection{Proof of Proposition~\ref{prop:conservative}}
\label{proof:conservative}

\begin{proof}[Proof of Proposition~\ref{prop:conservative}]
    Since $V$ is fixed, the nonconformity scores $V_1, \dots, V_n$ and $V_{n+j}$ are exchangeable. By a standard result from conformal inference~\citet[Proposition~2.4]{Vovk2005}, the oracle $p$-value $p_j^*$ defined as in (\ref{eq:oracle_p}) is uniform distributed with value ranging in $(0,1)$, and $\mathbb{P}(p_j^* \leq \alpha) \leq \alpha$. This gives
    \begin{align*}
        \mathbb{P}(p_j^* \leq \alpha \text{ and } j \in \mathcal{H}_0) \leq \alpha.
    \end{align*}
    When the null hypothesis $H_{0j}$ is true, $\boldsymbol{y}_{n+j} \in R^c$. Since $\boldsymbol{r}_{n+j} \in R$, by the regional monotone property, we have $V_{n+j} = V(\boldsymbol{x}_{n+j}, \boldsymbol{y}_{n+j}) \leq V(\boldsymbol{x}_{n+j}, \boldsymbol{r}_{n+j}) = \widehat{V}_{n+j}$. We then have $p_j^* \leq p_j$ by definition, and
    \begin{align*} 
        \mathbb{P}(p_j \leq \alpha \text{ and } j \in \mathcal{H}_0) \leq \mathbb{P}(p_j^* \leq \alpha \text{ and } j \in \mathcal{H}_0) \leq \alpha.
    \end{align*}
\end{proof}

\subsection{Proof of Theorem~\ref{thm:main}}
\label{proof:main}
\begin{proof}[Proof of Theorem~\ref{thm:main}]
    We adapt the proof of Theorem~6 in \cite{jin2023selection}.  In the proof that follows, we fix index $j \in \{1, \dots, m\}$.
    For notational simplicity, only in this proof we deal with the \textit{deterministic} conformal $p$-values, 
    \[ p_j = \frac{1}{n+1} \big[ 1 + \sum_{i=1}^n \mathbbm{1} \{ V_i \leq \widehat{V}_{n+j} \}\big] \]
    We note that the deterministic conformal $p$-values are only valid when the scores $\{V_i\}_{i=1}^n \cup\{V_{n+j}\}_{j=1}^m$ have no ties almost surely, and therefore we also make this assumption. We highlight that the validity of the \textit{random} conformal $p$-value does not rely on this statement.
    We also define the corresponding deterministic oracle $p$-values, 
    \[ p_j^* = \frac{1}{n+1} \big[ 1 + \sum_{i=1}^n \mathbbm{1} \{ V_i \leq V_{n+j} \}\big]. \]
    This version of $p$-values are also conservative by standard results in conformal inference~\citep{Bates2023, jin2023selection}.
    For $l = 1, \dots, m$, we define a set of slightly modified $p$-values,
    \[ p_l^{(j)} = \frac{1}{n+1} \big[ \sum_{i=1}^n \mathbbm{1} \{ V_i \leq \widehat{V}_{n+l} \} + \mathbbm{1}\{ V_{n+j} \leq \widehat{V}_{n+l} \} \big]. \]
    Also define $\mathcal{S}(a_1, \dots, a_m) \subseteq \{ 1, \dots, m\}$ as the rejection index set obtained by the Benjamini-Hochberg procedure, from $p$-values taking values $a_1, \dots, a_m$.  Then, the output of mCS is $\mathcal{S}(p_1, \dots, p_m)$. 
    
    In the sequel, we will compare $\mathcal{S}(p_1, \dots, p_m)$ to 
    \[ \mathcal{S}(p_1^{(j)}, \dots, p_{j-1}^{(j)}, p_j^*, p_{j+1}^{(j)}, \dots, p_m^{(j)}) \]
    for the test sample $j$ that is falsely rejected, i.e. 
    \[ j \in \mathcal{S}, Y_{n+j} \in R^c. \]
    First, on this event, we have $V_{n+j} = V(\boldsymbol{x}_{n+j}, \boldsymbol{y}_{n+j}) \leq V(\boldsymbol{x}_{n+j}, \boldsymbol{r}_{n+j}) = \widehat{V}_{n+j}$ and $p_j^* \leq p_j$. For the remaining $p$-values, $\{ p_l^{(j)} \}_{l \neq j}$, we consider two cases:
    \begin{enumerate}
        \item[$($i$)$] If $\widehat{V}_{n+l} \geq \widehat{V}_{n+j}$, then $p_l \geq p_j$. In this case, we also have $\widehat{V}_{n+l} \geq V_{n+j}$ as $\widehat{V}_{n+j} \geq V_{n+j}$. This implies
        \[ p_l^{(j)} = \frac{1}{n+1}\big[ \sum_{i=1}^n \mathbbm{1}\{V_i \leq \widehat{V}_{n+l} \} + \mathbbm{1}\{V_{n+j} \leq \widehat{V}_{n+l}\} \big] = \frac{1}{n+1}\big[ \sum_{i=1}^n \mathbbm{1}\{V_i \leq \widehat{V}_{n+l} \} + 1 \big] =  p_l. \]
        \item[$($ii$)$] If $\widehat{V}_{n+l} < \widehat{V}_{n+j}$, then $p_l \leq p_j$. We also have
        \[ p_l^{(j)} \leq \frac{1}{n+1}\big[ \sum_{i=1}^n \mathbbm{1}\{V_i \leq \widehat{V}_{n+l} \} +1  \big] \leq \frac{1}{n+1}\big[ \sum_{i=1}^n \mathbbm{1}\{V_i \leq \widehat{V}_{n+j} \} + 1 \big] = p_j. \]
        Since $j \in \mathcal{S}$, by the definition of Benjamini-Hochberg procedure $l \in \mathcal{S}$ as $p_l$ has smaller rank when ordering the $p$-values.    \end{enumerate}
    To summarize, if we replace $(p_1, \dots, p_m)$ by $(p_1^{(j)}, \dots, p_{j-1}^{(j)}, p_j^*, p_{j+1}^{(j)}, \dots, p_m^{(j)})$ on the event $\{ j \in \mathcal{S}, Y_{n+j} \in R^c\}$, such a replacement does not modify any of those $p$-values $p_l$ if they satisfied $p_l \geq p_j$. Also, for all $p$-values $p_l$ with $p_l \leq p_j$ including $j$ itself ($l=j$), their replaced values $p_l^{(j)}$ are still no greater than $p_j$. 
    Since all $p$-values are no larger than their original values after the replacements, the size of rejection set must not decrease. On the other hand, since $j \in \mathcal{S}$ and no $p$-values larger than $p_j$ are modified, no new hypotheses can be rejected by the new set of $p$-values. We conclude that
    \begin{align*}
       \mathcal{S}_j^* := \mathcal{S}(p_1^{(j)}, \dots, p_{j-1}^{(j)}, p_j^*, p_{j+1}^{(j)}, \dots, p_m^{(j)}) = \mathcal{S}(p_1, \dots, p_m) = \mathcal{S}
    \end{align*}
    on the event $\{ Y_{n+j} \in R^c, j \in \mathcal{S} \}$. By decomposing the FDR we have
    \begin{align*}
        \text{FDR} &= \mathbb{E}\bigg[ \frac{\sum_{j=1}^{m} \mathbbm{1} \{ Y_{n+j} \in R^c\} \mathbbm{1}\{j \in \mathcal{S}\}}{\max(1, |\mathcal{S}|)} \bigg] \\
        &= \sum_{j=1}^{m} \sum_{k=1}^{m} \frac{1}{k} \mathbb{E}\big[ \mathbbm{1}\{|\mathcal{S}| = k\} \mathbbm{1}\{Y_{n+j}\in R^c\} \mathbbm{1}\{p_j \leq \frac{qk}{m}, j \in \mathcal{S} \}  \big] \\
        &\leq \sum_{j=1}^{m} \sum_{k=1}^{m} \frac{1}{k} \mathbb{E}\big[ \mathbbm{1}\{|\mathcal{S}_j^*| = k\} \mathbbm{1}\{Y_{n+j}\in R^c\} \mathbbm{1}\{p_j^* \leq \frac{qk}{m} \}  \big]  \\
        &\leq \sum_{j=1}^{m} \sum_{k=1}^{m} \frac{1}{k} \mathbb{E}\big[ \mathbbm{1}\{|\mathcal{S}_j^*| = k\} \mathbbm{1}\{p_j^* \leq \frac{qk}{m} \} \big] \\
        &= \sum_{j=1}^{m} \sum_{k=1}^{m} \frac{1}{k} \mathbb{E}\big[ \mathbbm{1}\{|\mathcal{S}_j^*| = k\} \mathbbm{1}\{j \in \mathcal{S}_j^* \} \big].
    \end{align*}
    The last equality is again by the property of the Benjamini-Hochberg procedure. Also, by its step-up nature, sending $p_j^*$ to zero does not change the rejection set if the corresponding hypothesis of $p_j^*$ is rejected, i.e. on the event $\{j \in \mathcal{S}_j^* \}$. We have
    \[ \mathcal{S}_j^* = \mathcal{S}(p_1^{(j)}, \dots, p_{j-1}^{(j)}, p_j^*, p_{j+1}^{(j)}, \dots, p_m^{(j)}) = \mathcal{S}(p_1^{(j)}, \dots, p_{j-1}^{(j)}, 0, p_{j+1}^{(j)}, \dots, p_m^{(j)}) =: \mathcal{S}_{j \rightarrow 0}^*  \]
    Thus,
    \begin{align*}
        \text{FDR} \leq \sum_{j=1}^{m} \sum_{k=1}^{m} \frac{1}{k} \mathbb{E}\big[ \mathbbm{1}\{|\mathcal{S}_{j \rightarrow 0}^*| = k\} \mathbbm{1}\{j \in \mathcal{S}_{j}^* \} \big] = \sum_{j=1}^m \mathbb{E}\bigg[ \frac{ \mathbbm{1} \{p_j^* \leq q|\mathcal{S}_{j}^*| / m \}}{\max(1, |\mathcal{S}_{j \rightarrow 0}^*|)} \bigg] \leq \sum_{j=1}^m \mathbb{E}\bigg[ \frac{ \mathbbm{1} \{p_j^* \leq q|\mathcal{S}_{j \rightarrow 0}^*| / m \}}{\max(1, |\mathcal{S}_{j \rightarrow 0}^*|)} \bigg]
    \end{align*}
    By definition, $\{ p_l^{(j)} \}_{l \neq j}$ is invariant after permuting $\{ V_i\}_{i=1}^n \cup \{ V_{n+j}\}$. Since $\{ V_i\}_{i=1}^n \cup \{ V_{n+j}\}$ are exchangeable conditioned on $\{ \widehat{V}_{n+l} \}_{l \neq j}$, the distribution of $\{ p_l^{(j)}\}_{l \neq j}$ is independent from the ordering of $\{ V_i\}_{i=1}^n \cup \{ V_{n+j}\}$ conditioned on the (unordered) set $[V_1, \dots, V_n, V_{n+j}] \cup \{\widehat{V}_{n+l}\}_{l \neq j}$. Also, conditioned on $\{\widehat{V}_{n+l}\}_{l \neq j}$, $\mathcal{S}_{j \rightarrow 0}^*$ only depends on $\{ p_l^{(j)}\}_{l \neq j}$ which is in turn only dependent on the unordered set $[V_1, \dots, V_n, V_{n+j}]$, and $p_j^*$ only depends on the ordering of $\{ V_i\}_{i=1}^n \cup \{ V_{n+j}\}$. This implies that $\mathcal{S}_{j \rightarrow 0}^*$ is independent on $p_j^*$ conditioned on the (unordered) set $[V_1, \dots, V_n, V_{n+j}]$ and $\{\widehat{V}_{n+l}\}_{l \neq j}$. Therefore, by the conservative property of conformal $p$-values and conditional independence, 
    \begin{align*}
        \mathbb{P} \Big(p_j^* \leq \frac{q|\mathcal{S}_{j \rightarrow 0}^*|}{m} \Bigm| [V_1, \dots, V_n, V_{n+j}] \cup \{\widehat{V}_{n+l}\}_{l \neq j} \Big) \leq \frac{ q|\mathcal{S}_{j \rightarrow 0}^*| }{m}.
    \end{align*}
    By the law of total expectation,
    \begin{align*}
        \mathbb{E}\bigg[ \frac{ \mathbbm{1} \{p_j^* \leq q|\mathcal{S}_{j \rightarrow 0}^*| / m \}}{\max(1, |\mathcal{S}_{j \rightarrow 0}^*|)} \bigg] 
        &= \mathbb{E}\Bigg[ \mathbb{E}\bigg[ \frac{ \mathbbm{1} \{p_j^* \leq q|\mathcal{S}_{j \rightarrow 0}^*| / m \}}{\max(1, |\mathcal{S}_{j \rightarrow 0}^*|)} \Bigm| [V_1, \dots, V_n, V_{n+j}] \cup \{\widehat{V}_{n+l}\}_{l \neq j} \bigg] \Bigg] \\
        &\leq \mathbb{E}\Big[ \mathbbm{1} \{ |\mathcal{S}_{j \rightarrow 0}^*| \neq 0 \} \frac{q|\mathcal{S}_{j \rightarrow 0}^*|}{m|\mathcal{S}_{j \rightarrow 0}^*|}  \Bigm| [V_1, \dots, V_n, V_{n+j}] \cup \{\widehat{V}_{n+l}\}_{l \neq j} \Big] \\
        &\leq \frac{q}{m}.
    \end{align*}
    Since when $|\mathcal{S}_{j \rightarrow 0}^*| = 0$, $\mathbbm{1} \{p_j^* \leq q|\mathcal{S}_{j \rightarrow 0}^*| / m \} = 0$. Now, the proof is concluded by summing over every $j = 1, \dots, m$.
\end{proof}

\subsection{Proof of Proposition~\ref{prop:expressive}}

\begin{proof}[Proof of Proposition~\ref{prop:expressive}]
\label{proof:expressive}
    For any specific dataset and selection problem, let $V^{\textnormal{opt}}$ denote the optimal nonconformity score that controls the FDR and maximize the number of selected samples. Define $V_i^{\textnormal{opt}} = V^{\textnormal{opt}}(\boldsymbol{x}_i, \boldsymbol{y}_i)$ for calibration samples and $\widehat{V}_{n+j}^{\textnormal{opt}} = V^{\textnormal{opt}}(\boldsymbol{x}_{n+j}, \boldsymbol{r}_{n+j})$ for test samples. Consider the following score $W$ within the family~\eqref{eq:score_class}:
    \begin{align*}
        W(\boldsymbol{x}, \boldsymbol{y}) = M \cdot\mathbbm{1}\{ \boldsymbol{y} \notin R^c \cup \partial R\} + V^{\textnormal{opt}}(\boldsymbol{x}, \boldsymbol{y}).
    \end{align*}
    Picking $\boldsymbol{r}_{n+j} \equiv \boldsymbol{r} \in \partial R$, the test nonconformity scores $\widehat{W}_{n+j}$ satisfy:
    \begin{align*}
        \widehat{W}_{n+j} = W(\boldsymbol{x}_{n+j}, \boldsymbol{r}_{n+j}) = V^{\textnormal{opt}}(\boldsymbol{x}_{n+j}, \boldsymbol{r}_{n+j}) = \widehat{V}_{n+j}^{\textnormal{opt}},
    \end{align*}
    and for each calibration score,
    \begin{align*}
        W_i = W(\boldsymbol{x}_i, \boldsymbol{y}_i) \geq V^{\textnormal{opt}}(\boldsymbol{x}_i, \boldsymbol{y}_i) = V_i^{\textnormal{opt}}.
    \end{align*}
    Therefore, by replacing $V^{\textnormal{opt}}$ with $W$, the calibration nonconformity scores $W_1,\ldots,W_n$ may increase, while the test nonconformity scores $\widehat{W}_{n+j}$ remain unchanged. According to~\eqref{eq:conf_p}, this leads to a decrease in each conformal $p$-value $p_j$, which in turn increase the value $k^*$ in Algorithm~\ref{algo:mcs}. This means that we would select more samples, by the definition of $\mathcal{S}$. Consequently, we conclude that there must exist a score within the family~\eqref{eq:score_class} that achieves the optimal selection size. A similar argument applies to the maximization selection power, thereby concluding the proof.
\end{proof}

\section{Deferred Discussions}
\label{deferred_discussions}

\subsection{Discussions on Classification Responses}
\label{sec:cls_discussion}

In our paper, we choose to focus primarily on the regression response case for multivariate conformal selection.
In fact, the selection problem for classification responses (univariate or multivariate) can be directly reduced to the univariate conformal selection framework introduced by~\citet{jin2023selection}:

For the univariate classification setting, suppose the response space is composed of classes $\mathcal{Y} = \cup_{k=1}^K C_k$ with target region $R = \cup_{k=1}^{s} C_k$ (with $s < K$). Then, by defining a binary response: $\tilde{y}_{i} = \mathbbm{1}\{y_{i} \in R \}$, the original selection problem directly translates into a univariate conformal selection problem, where we select samples with $\tilde{y}_{i} = 1$.

For multivariate classification case, e.g.  suppose that bivariate responses $\boldsymbol{y}_{i}$ are drawn from a joint class space $\mathcal{Y} = (\mathcal{Y}^{(1)},\mathcal{Y}^{(2)}) = \cup_{k,\ell} (C^{(1)}_k, C^{(2)}_\ell)$, and the target region $R = \cup_{(k, \ell) \in \mathcal{I}}(C^{(1)}_k, C^{(2)}_\ell)$. Again, we define a binary indicator $\tilde{y}_{i} = \mathbbm{1}\{\boldsymbol{y}_{i} \in R \},$ 
converting the original multivariate selection problem into a standard univariate selection task: 
$$H_{0j}: \tilde{y}_{n+j} < 0.5 \text{\quad versus \quad } H_{1j}: \tilde{y}_{n+j} \geq 0.5.$$  

Since $P(\tilde{y}_{i} =1)=P(\boldsymbol{y}_{i} \in R)$, there is a direct correspondence between the multivariate and univariate nonconformity scores:
\[
V(\boldsymbol{x},\boldsymbol{y}) = M\cdot\mathbbm{1}\{\boldsymbol{y}\in R\} - \widehat{P}(\boldsymbol{y}\in R|\boldsymbol{x}) 
\]
and
\[
V(\boldsymbol{x},\tilde{y}) = M\cdot\mathbbm{1}\{\tilde{y}\geq0.5\} - \tilde{\mu}(\boldsymbol{x}) 
\]
where $\tilde{\mu}(\boldsymbol{x})\equiv \widehat{P}(\tilde{y}=1|\boldsymbol{x})$. Moreover, regional monotonicity is simply the usual monotone condition of univariate conformal selection:
$V(\boldsymbol{x},\tilde{y}=0)\leq V(\boldsymbol{x},\tilde{y}=1).$

Thus, every classification-based selection task can be naturally and effectively solved using existing univariate conformal selection methods.  

\subsection{Additional Discussions on Theorem~\ref{thm:heuristic} and Advantages of
The Clipped Nonconformity Score}
\label{sec:adv_clip}

In this section, we first provide interpretations about Theorem~\ref{thm:heuristic} and offer an alternative perspective on the advantages of the clipped nonconformity score.


We first note that the (practical) conformal $p$-values defined in~\eqref{eq:conf_p} can be rewritten as
\begin{align*}
    p_j &= \frac{1}{n+1} \sum_{i=1}^n\mathbbm{1}\{V_i < \widehat{V}_{n+j} \} + \frac{1}{n+1} U_j \sum_{i=1}^n \mathbbm{1}\{V_i = \widehat{V}_{n+j}\} + \frac{1}{n+1} U_j .
\end{align*}
By the Glivenko-Cantelli theorem (the uniform strong law of large numbers), as $n\rightarrow\infty$,
\begin{align*}
    \sup_{t\in\mathbb{R}} \left| \frac{1}{n+1} \sum_{i=1}^n\mathbbm{1}\{V_i < t \} + \frac{1}{n+1} U_j \sum_{i=1}^n \mathbbm{1}\{V_i = t\} + \frac{1}{n+1} U_j - \mathbb{P}(V(\boldsymbol{x}, \boldsymbol{y}) < t ) - U_j\cdot \mathbb{P}(V(\boldsymbol{x}, \boldsymbol{y}) = t) \right| \overset{a.s.}{\rightarrow} 0.
\end{align*}
Replacing $t$ in the above by $\widehat{V}_{n+j}=V(\boldsymbol{x}_{n+j}, \boldsymbol{r})$ yields
\begin{align*}
    p_j &\overset{a.s.}{\rightarrow} \mathbb{P}(V(\boldsymbol{x}, \boldsymbol{y}) < V(\boldsymbol{x}_{n+j}, \boldsymbol{r}) | V(\boldsymbol{x}_{n+j}, \boldsymbol{r}) ) + U_j\cdot \mathbb{P}(V(\boldsymbol{x}, \boldsymbol{y}) = V(\boldsymbol{x}_{n+j}, \boldsymbol{r}) | V(\boldsymbol{x}_{n+j}, \boldsymbol{r}) )
    = F(V(\boldsymbol{x}_{n+j}, \boldsymbol{r}), U_j) \\
    & \overset{d}{\sim} F(V(\boldsymbol{x}, \boldsymbol{r}), U).
\end{align*}
Therefore, $F(V(\boldsymbol{x}, \boldsymbol{r}), U)$ is also a conservative random variable.

The quantity $t^*$ in Theorem~\ref{thm:heuristic} can also be viewed as the asymptotic counterpart of a corresponding finite-sample quantity. A characterization of the BH procedure~\citep{storey2004strong} states that the rejection set $\mathcal{S}=\{j: p_j \leq \tau^*\}$,  where the BH rejection threshold $\tau^*$ is defined as
\begin{align*}
    \tau^* = \sup \Big\{ t\in[0,1]: \frac{t}{\frac{1}{m}\sum_{j=1}^m \mathbbm{1}\{p_j \leq t\}} \leq q \Big\}.
\end{align*}
By the fact that $\frac{1}{m}\sum_{j=1}^m \mathbbm{1}\{p_j \leq t\} - \mathbb{P}(p_j \leq t) \overset{p}{\rightarrow} 0$ as $m\rightarrow \infty$ (which follows due to the the weak law of large numbers for triangular arrays), and our earlier characterization of $F(V(\boldsymbol{x}, \boldsymbol{r}), U_j)$ which implies $\mathbb{P}(p_j \leq t) \rightarrow \mathbb{P}( F(V(\boldsymbol{x}, \boldsymbol{r}), U) \leq t)$ as $n\rightarrow \infty$, the fraction in the definition of $\tau^*$ satisfies (in the limit $n, m\rightarrow \infty$)
\begin{align*}
    \frac{t}{\frac{1}{m}\sum_{j=1}^m \mathbbm{1}\{p_j \leq t\}} \overset{p}{\rightarrow} 
    \frac{t}{\mathbb{P}(F(V(\boldsymbol{x}, \boldsymbol{r}), U) \leq t)}.
\end{align*}
This establishes $t^*$ as the asymptotic limit of the BH rejection threshold $\tau^*$, which further implies
\begin{align*}
    \frac{1}{m}\sum_{j=1}^m \mathbbm{1}\{p_j \leq \tau^*\} \overset{p}{\rightarrow} \mathbb{P}(F(V(\boldsymbol{x}, \boldsymbol{r}), U) \leq t^*) ~~\text{and}~~ \frac{1}{m} \sum_{j=1}^m \mathbbm{1}\{p_j \leq \tau^* \textnormal{ and } \boldsymbol{y}_{n+j} \in R\} \overset{p}{\rightarrow} \mathbb{P}(F(V(\boldsymbol{x}, \boldsymbol{r}), U) \leq t^*, \boldsymbol{y} \in R) .
\end{align*}
Theorem~\ref{thm:heuristic} then gives the asymptotic version of FDR and power:
\begin{align*}
    \text{FDR} &= \mathbb{E}\Bigg[\frac{|\mathcal{S} \cap \{j: \boldsymbol{y}_{n+j} \in R^c\}|}{|\mathcal{S}|} \Bigg]
    = \mathbb{E}\Bigg[\frac{|\{j: p_j \leq \tau^*\} \cap \{j: \boldsymbol{y}_{n+j} \in R^c\}|}{|\{j: p_j \leq \tau^*\}|} \Bigg] \\
    &= \mathbb{E}\Bigg[\frac{ \frac{1}{m} \sum_{j=1}^m \mathbbm{1}\{p_j \leq \tau^* \textnormal{ and } \boldsymbol{y}_{n+j} \in R^c\} }{\frac{1}{m} \sum_{j=1}^m \mathbbm{1}\{p_j \leq \tau^*\}} \Bigg] \\
    & 
    \rightarrow \frac{\mathbb{P}(F(V(\boldsymbol{x}, \boldsymbol{r}), U) \leq t^*, \boldsymbol{y} \in R^c)}{\mathbb{P}(F(V(\boldsymbol{x}, \boldsymbol{r}), U) \leq t^*)} ,
\end{align*}
\begin{align*}
    \text{Power} &= \mathbb{E}\Bigg[\frac{|\mathcal{S} \cap \{j: \boldsymbol{y}_{n+j} \in R\}|}{|\{j: \boldsymbol{y}_{n+j} \in R\}|} \Bigg]
    = \mathbb{E}\Bigg[\frac{|\{j: p_j \leq \tau^*\} \cap \{j: \boldsymbol{y}_{n+j} \in R\}|}{|\{j: \boldsymbol{y}_{n+j} \in R\}|} \Bigg] \\
    &= \mathbb{E}\Bigg[\frac{ \frac{1}{m} \sum_{j=1}^m \mathbbm{1}\{p_j \leq \tau^* \textnormal{ and } \boldsymbol{y}_{n+j} \in R\} }{\frac{1}{m} \sum_{j=1}^m \mathbbm{1}\{ \boldsymbol{y}_{n+j} \in R \}} \Bigg] \\
    &
    \rightarrow \frac{\mathbb{P}(F(V(\boldsymbol{x}, \boldsymbol{r}), U) \leq t^*, \boldsymbol{y} \in R)}{\mathbb{P}(\boldsymbol{y} \in R)}.
\end{align*}

Theorem~\ref{thm:heuristic} suggests that the clipped score should be preferred from the perspective of maximizing the realized (asymptotic) FDR (which is still controlled below the nominal level). The asymptotic FDR is bounded through the following chain of inequalities:
\begin{align} \label{eq:inequalities}
    \frac{\mathbb{P}(F(V(\boldsymbol{x}, \boldsymbol{r}), U) \leq t^*, \boldsymbol{y} \in R^c)}{\mathbb{P}(F(V(\boldsymbol{x}, \boldsymbol{r}), U) \leq t^*)} \nonumber
    &\leq \frac{\mathbb{P}(F(V(\boldsymbol{x}, \boldsymbol{y}), U) \leq t^*, \boldsymbol{y} \in R^c)}{\mathbb{P}(F(V(\boldsymbol{x}, \boldsymbol{r}), U) \leq t^*)}  \\
    &\overset{\textnormal{($*$)}}{\leq} \frac{\mathbb{P}(F(V(\boldsymbol{x}, \boldsymbol{y}), U) \leq t^*)}{\mathbb{P}(F(V(\boldsymbol{x}, \boldsymbol{r}), U) \leq t^*)}
    = \frac{t^*}{\mathbb{P}(F(V(\boldsymbol{x}, \boldsymbol{r}), U) \leq t^*)} .
\end{align}
By the definition of $t^*$, the final term in \eqref{eq:inequalities} is at most the nominal level $q$. If the bounds in this chain of inequalities could be further tightened, the realized FDR would more closely align with the nominal level $q$, effectively allowing for greater selection power as more of the FDR budget could be utilized.

Using the clipped score tightens the inequality ($*$) in \eqref{eq:inequalities}. Assuming the score function $V$ is clipped, we observe that
\begin{align*}
    \frac{\mathbb{P}(F(V(\boldsymbol{x}, \boldsymbol{y}), U) \leq t^*)}{\mathbb{P}(F(V(\boldsymbol{x}, \boldsymbol{r}), U) \leq t^*)} - \frac{\mathbb{P}(F(V(\boldsymbol{x}, \boldsymbol{y}), U) \leq t^*, \boldsymbol{y} \in R^c)}{\mathbb{P}(F(V(\boldsymbol{x}, \boldsymbol{r}), U) \leq t^*)} = \frac{\mathbb{P}(F(V(\boldsymbol{x}, \boldsymbol{y}), U) \leq t^*, \boldsymbol{y} \in R)}{\mathbb{P}(F(V(\boldsymbol{x}, \boldsymbol{r}), U) \leq t^*)}
\end{align*}
is approximately zero, since for $\boldsymbol{y} \in R$, the clipped score satisfies $V(\boldsymbol{x}, \boldsymbol{y}) \approx M$, implying that $\mathbb{P}(F(V(\boldsymbol{x}, \boldsymbol{y}), U) \leq t^*, \boldsymbol{y} \in R) \approx 0$, because that $F(v,u)$ is monotone with
respect to its first argument $v$. Notably, the large constant $M$ is originally introduced in conformal selection as a relaxation of infinity, and these approximations indeed hold. 
This phenomenon is verified through our extra simulations in Appendix~\ref{extra_experiments:score_comp}.

\section{Additional Details for Numerical Simulations (Section~\ref{sec:simu})}
\subsection{Data Generating Processes and Configuration of the Selection Tasks}
\label{sec:data_gen}

The configuration of the true regression function $\mu$ and noise term $\boldsymbol{\epsilon}\in\mathbb{R}^d$ can be found in Table~\ref{tab:true_regr}.

For the noise term, the degree of freedom for the $t$-distribution scenario is set to $\nu = 3$, and the scale matrix $\Sigma$ is specified as a matrix with diagonal entries of 0.5 and off-diagonal entries of 0.05. The second column denotes the $k$-th output entry of the true regression function $\mu$, which relates to the $k$-th coordinate $y_k$ of the response $\boldsymbol{y}$ and the $k$-th coordinate $\epsilon_k$ of the noise $\boldsymbol{\epsilon}$ as $y_k = [\mu(\boldsymbol{x})]_k + \epsilon_k$. In all of our simulations, we take $p=10$ (recall that $\boldsymbol{x}\sim\text{Unif}(-1, 1)^p$ but $\boldsymbol{y}\in\mathbb{R}^d$). If, when generating $y_k$ the value of $x_\ell$ (the $\ell$-th coordinate of $\boldsymbol{x}$) for $\ell>p$ is needed, we take $x_\ell = x_{((\ell-1)~\text{mod}~p)+1}$.


Table~\ref{tab:coeff} summarizes the values of coefficient vectors ($\boldsymbol{c}$ and $\boldsymbol{r}$) that define the selection tasks for each response dimension. 
Within each vector, all entries share the same value. For instance, if $c_k$ is listed as 1, it indicates that $\boldsymbol{c} = (1, 1, \dots, 1)$.

The six settings we consider differ in two key aspects that influence the difficulty of selection. First, the regression function $\mu$ exhibits varying degrees of nonlinearity across settings. Specifically, Settings 1 and 4 are linear, while Settings 2 and 5 have weak nonlinearity, and Settings 3 and 6 exhibit strong nonlinearity. The degree of nonlinearity affects the predictive accuracy of the estimated regression function $\hat\mu$, which in turn impacts the selection difficulty.
Second, the distribution of the noise term $\boldsymbol{\epsilon}$ differs across settings, leading to variations in the conditional variance. In Settings 1, 2, and 3, $\boldsymbol{\epsilon}$ follows a multivariate Gaussian distribution, whereas in Settings 4, 5, and 6, it follows a multivariate $t$-distribution. The conditional variance $\textnormal{cov}(\boldsymbol{\epsilon})=\textnormal{Var}(\boldsymbol{y} \mid \boldsymbol{x})$ in the latter three settings is $\frac{\nu}{\nu - 2} \Sigma = 3\Sigma$, which is three times larger than that of the former three settings, thereby increasing the difficulty of selection.

\begin{table}[!h]
\centering
\caption{True regression functions and noise distributions}
\label{tab:true_regr}
\setlength{\tabcolsep}{12pt}
\renewcommand{\arraystretch}{2.3}
\begin{tabular}{ c  c  c }
\toprule  
Setting       & $[\mu(\cdot)]_k$ & $\boldsymbol{\epsilon}$ \\
\midrule
1     & $x_{k} - \frac{1}{2} x_{k+1} + x_{k+2} + \frac{3}{2}$  & $\mathcal{N}(0, \Sigma)$  \\  
2     & $x_{k} + x_{k+2}^2 + \frac{1}{2}$  
& $\mathcal{N}(0, \Sigma)$ \\   
3     & \makecell{$\mathbbm{1}\{x_k x_{k+1} > 0\} \cdot \mathbbm{1}\{x_{k+2} > 0.5\} \cdot (0.25 + x_{k+2})$ \\ $\mathbbm{1}\{x_k x_{k+1} \leq 0\} \cdot \mathbbm{1}\{x_{k+2} \leq 0.5\} \cdot (x_{k+2} - 0.25) + 0.75$} & $\mathcal{N}(0, \Sigma)$  \\
4     & \text{Same as Setting~1}  & $t_\nu(0, \Sigma)$  \\  
5     & \text{Same as Setting~2} & $t_\nu(0, \Sigma)$   \\  
6     & \text{Same as Setting~3} & $t_\nu(0, \Sigma)$\\ 
\bottomrule
\end{tabular}
\end{table}

\begin{table}[!h]
\centering
\caption{Coefficients for Selection Task 1 and 2}
\label{tab:coeff}
\setlength{\tabcolsep}{12pt}
\renewcommand{\arraystretch}{1.5}
\begin{tabular}{cccc}
\toprule  
Response Dimension       & $c_k$ (Task 1) & $c_k$ (Task 2) & $r_k$ (Task 2) \\
\midrule
$2$ & $1$ & $2$ & $1.5$ \\
$5$ & $0.2$ & $2$ & $2.6$ \\
$10$ & $-0.2$ & $2$ & $4.1$ \\
$30$ & $-0.6$ & $2$ & $7.5$ \\
\bottomrule
\end{tabular}
\end{table}

To better illustrate the data distributions, Figure~\ref{fig:data_vis} presents example scatter plots of the response vector $\boldsymbol{y}$ for the six settings, with the response dimension set to 2 for visualization purposes. Higher-dimensional cases are not displayed as they are more challenging to interpret.
Settings within the same column share the same conditional expectation $\mathbb{E}(\boldsymbol{y} \mid \boldsymbol{x})$, while those in the lower row exhibit greater noise and more dispersed scatter patterns due to the multivariate $t$-distributed noise $\boldsymbol{\epsilon}$. The target regions for the two tasks, defined based on the coefficients in Table~\ref{tab:coeff}, are highlighted as red and yellow shaded areas in the plots.

\begin{figure}[!h]
    \centering
    \includegraphics[width=0.99\textwidth]{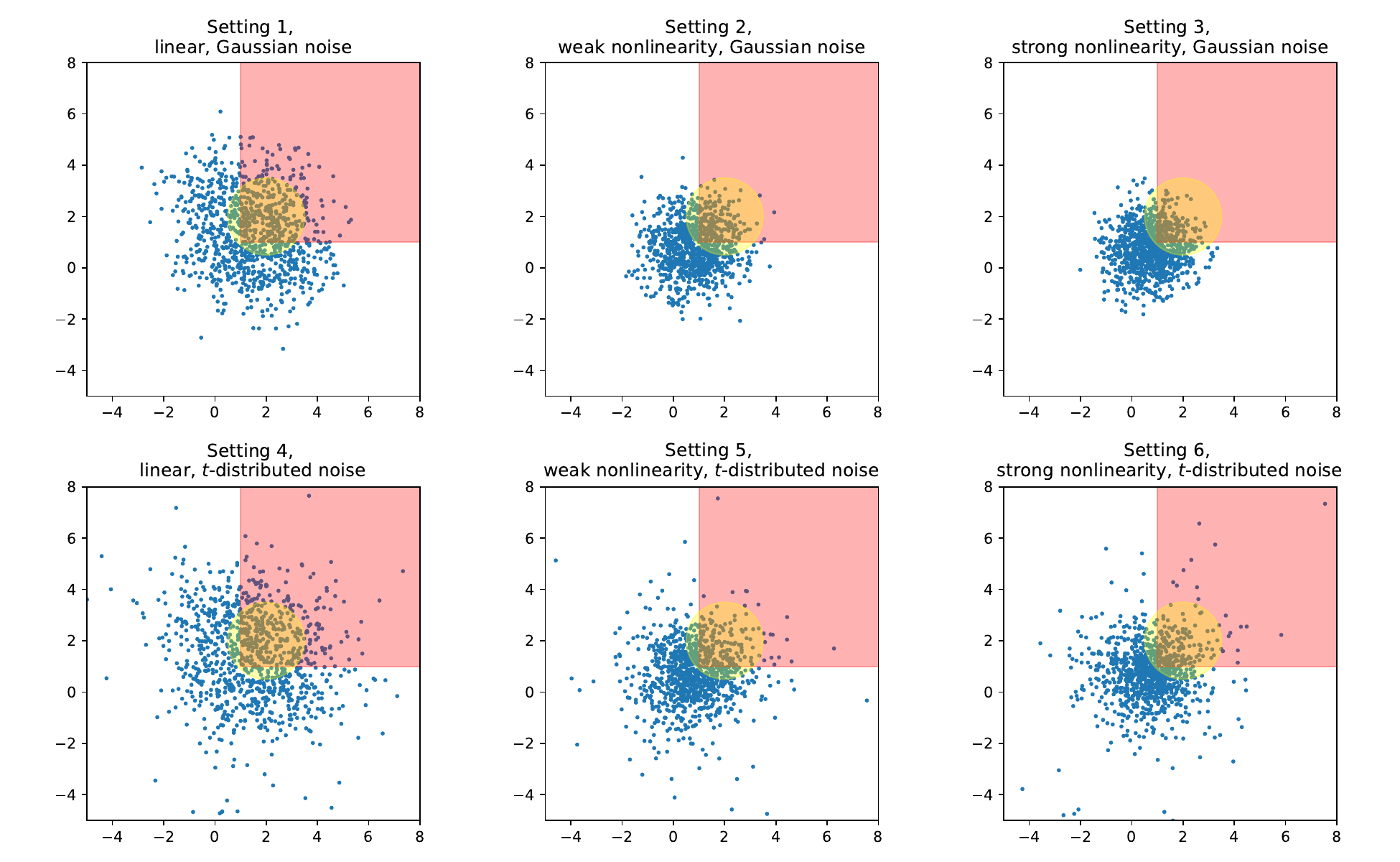}  
    \caption{Scatter plots of 1,000 i.i.d.\,samples from our data generating processes. Each point represents a sample, with the $x$- and $y$-axes corresponding to its responses entries $\boldsymbol{y}_1$ and $\boldsymbol{y}_2$, respectively. The red and the yellow shaded areas represents the two target regions for the two selection tasks. }
    \label{fig:data_vis}  
\end{figure}


\subsection{Extra Simulated Experiments} \label{extra_experiments}

\subsubsection{Comparing Two Distance-Based Scores} \label{extra_experiments:score_comp}

In Tables~\ref{tab:dist_compare_1} and \ref{tab:dist_compare_2}, we compare the performance of the two distance-based scores \eqref{eq:dist_score_1} and \eqref{eq:dist_score_2} in \texttt{mCS-dist}. All other configurations are kept the same as in Section~\ref{sec:simu}.

\begin{table}[!h]
\centering
\small
\setlength{\tabcolsep}{12pt}
\renewcommand{\arraystretch}{1.2}
\caption{Observed FDR of \texttt{mCS-dist} with score \eqref{eq:dist_score_1} and \eqref{eq:dist_score_2} for Task 1 and 2.}
\label{tab:dist_compare_1}

\begin{tabular}{lcccc}
\toprule
 & \multicolumn{2}{c}{Task 1} & \multicolumn{2}{c}{Task 2} \\
\cmidrule(lr){2-3} \cmidrule(lr){4-5}
Setting & Score \eqref{eq:dist_score_1} & Clipped Score \eqref{eq:dist_score_2} & Score \eqref{eq:dist_score_1} & Clipped Score \eqref{eq:dist_score_2} \\
\midrule
1 & $0.154$ & $0.277$ & $0.102$ & $0.265$ \\
2 & $0.269$ & $0.314$ & $0.127$ & $0.263$ \\
3 & $0.241$ & $0.265$ & $0.150$ & $0.223$ \\
4 & $0.202$ & $0.278$ & $0.226$ & $0.286$ \\
5 & $0.282$ & $0.308$ & $0.240$ & $0.292$ \\
6 & $0.265$ & $0.287$ & $0.220$ & $0.283$ \\
\bottomrule
\end{tabular}
\end{table}

\begin{table}[!h]
\centering
\small
\setlength{\tabcolsep}{12pt}
\renewcommand{\arraystretch}{1.2}
\caption{Observed power of \texttt{mCS-dist} with score \eqref{eq:dist_score_1} and \eqref{eq:dist_score_2} for Task 1 and 2.}
\label{tab:dist_compare_2}

\begin{tabular}{lcccc}
\toprule
 & \multicolumn{2}{c}{Task 1} & \multicolumn{2}{c}{Task 2} \\
\cmidrule(lr){2-3} \cmidrule(lr){4-5}
Setting & Score \eqref{eq:dist_score_1} & Clipped Score \eqref{eq:dist_score_2} & Score \eqref{eq:dist_score_1} & Clipped Score \eqref{eq:dist_score_2} \\
\midrule
1 & $0.305$ & $\color{red}0.555$ & $0.440$ & $\color{red}0.760$ \\
2 & $0.069$ & $\color{red}0.104$ & $0.208$ & $\color{red}0.405$ \\
3 & $0.046$ & $\color{red}0.068$ & $0.075$ & $\color{red}0.134$ \\
4 & $0.191$ & $\color{red}0.324$ & $0.222$ & $\color{red}0.333$ \\
5 & $0.052$ & $\color{red}0.060$ & $0.093$ & $\color{red}0.170$ \\
6 & $0.038$ & $\color{red}0.046$ & $0.048$ & $\color{red}0.063$ \\
\bottomrule
\end{tabular}
\end{table}

Across a wide range of tasks and experimental settings, the clipped score \eqref{eq:dist_score_2} consistently demonstrates superior performance. Although both scoring methods guarantee valid FDR control, the clipped score tends to exhibit a higher observed FDR. This finding is consistent with the asymptotic analysis in Appendix~\ref{sec:adv_clip}, which suggests that the clipped score can leverage a larger portion of the available FDR budget.

\subsubsection{Varying Response Dimension}

The main article included results with response dimension $d=30$.  In this section, we access the performance of \texttt{mCS-dist} and \texttt{mCS-learn} in lower dimensional response settings $d \in \{2, 5, 10\}$, examining whether \texttt{mCS-learn} continues to outperform competing methods for smaller values of $d$. 
To simplify the analysis, we focus on Setting 3 for these experiments. All other configurations are kept unchanged.

\begin{table}[!h]
\centering
\small
\setlength{\tabcolsep}{12pt}
\renewcommand{\arraystretch}{1.2}
\caption{Observed FDR of different methods for lower response dimensions.}
\label{tab:other_dim_fdp}

\begin{tabular}{lccccccccc}
\toprule
 & \multicolumn{6}{c}{Task 1} & \multicolumn{3}{c}{Task 2} \\
\cmidrule(lr){2-7} \cmidrule(lr){8-10}
$d$ & \texttt{CS\_int} & \texttt{CS\_ib} & \texttt{CS\_is} & \texttt{bi} & \texttt{mCS-d} & \texttt{mCS-l} & \texttt{bi} & \texttt{mCS-d} & \texttt{mCS-l} \\
\midrule
2 & $0.244$ & $0.000$ & $0.142$ & $0.303$ & $0.268$ & $0.281$ & $0.270$ & $0.299$ & $0.302$ \\
5 & $0.706$ & $0.000$ & $0.255$ & $0.290$ & $0.324$ & $0.294$ & $0.246$ & $0.290$ & $0.302$ \\
10 & $0.735$ & $0.000$ & $0.342$ & $0.311$ & $0.276$ & $0.310$ & $0.257$ & $0.265$ & $0.311$ \\
30 & $0.724$ & $0.000$ & $0.204$ & $0.295$ & $0.264$ & $0.278$ & $0.216$ & $0.223$ & $0.263$ \\
\bottomrule
\end{tabular}
\end{table}

\begin{table}[!h]
\centering
\small
\setlength{\tabcolsep}{12pt}
\renewcommand{\arraystretch}{1.2}
\caption{Observed power of different methods for lower response dimensions.}
\label{tab:other_dim_power}

\begin{tabular}{lccccccccc}
\toprule
 & \multicolumn{6}{c}{Task 1} & \multicolumn{3}{c}{Task 2} \\
\cmidrule(lr){2-7} \cmidrule(lr){8-10}
$d$ & \texttt{CS\_int} & \texttt{CS\_ib} & \texttt{CS\_is} & \texttt{bi} & \texttt{mCS-d} & \texttt{mCS-l} & \texttt{bi} & \texttt{mCS-d} & \texttt{mCS-l} \\
\midrule
2 & $0.029$ & $0.000$ & $0.022$ & $0.032$ & $0.047$ & $\color{red}{0.063}$ & $0.050$ & $0.068$ & $\color{red}{0.069}$ \\
5 & $1.000$ & $0.000$ & $0.067$ & $0.041$ & $0.049$ & $\color{red}{0.083}$ & $0.050$ & $0.057$ & $\color{red}{0.072}$ \\
10 & $1.000$ & $0.000$ & $0.055$ & $0.044$ & $0.058$ & $\color{red}{0.069}$ & $0.062$ & $0.069$ & $\color{red}{0.106}$ \\
30 & $1.000$ & $0.000$ & $0.039$ & $0.059$ & $0.068$ & $\color{red}{0.102}$ & $0.115$ & $0.134$ & $\color{red}{0.179}$ \\
\bottomrule
\end{tabular}
\end{table}

Tables~\ref{tab:other_dim_fdp} and \ref{tab:other_dim_power} show that even for smaller $d$, \texttt{mCS-learn} and \texttt{mCS-dist} continue to outperform the other baseline methods, achieving the best and the second-best power respectively (while controlling FDR).
Notably, the performance gap of \texttt{mCS-learn} over competing methods becomes more pronounced as the response dimension increases, suggesting that \texttt{mCS-learn} is particularly advantageous for high-dimensional settings.



\subsubsection{Nonconvex Target Region} \label{nonconvex}

To examine the suitability of our methods for tasks involving more irregular target regions, we conduct experiments on two additional tasks in which the target region $R$ is nonconvex:
\begin{enumerate}[start=3, itemsep=0.3pt, topsep=1pt]
    \item The complement of the (shifted) orthant, $R = \{\boldsymbol{y}: y_k \geq c_k \hspace{0.2cm} \forall k\}^c = \{\boldsymbol{y}: y_k < c_k \hspace{0.2cm} \textnormal{for some } k\}$,
    \item The complement of a sphere centered at $\boldsymbol{c}$, $R = \{\boldsymbol{y}: \|\boldsymbol{y} - \boldsymbol{c}\|_2 > r\}$.
\end{enumerate}

To ensure that a reasonable proportion (15\%-35\%) of responses $\boldsymbol{y}$ fall within the selection region, the coefficients $c_k$ and $r$ are defined differently from their counterparts in Task~1 and 2. Table~\ref{tab:coeff_34} presents the specific values, where each scalar, following the convention in Table~\ref{tab:coeff}, represents a vector whose entries are all equal to that scalar.
All other setups, configurations, and model hyperparameters are the same as in Section~\ref{subsec:simu}.

\begin{table}[!h]
\centering
\caption{Coefficients for Selection Task 3 and 4}
\label{tab:coeff_34}
\setlength{\tabcolsep}{12pt}
\renewcommand{\arraystretch}{1.5}
\begin{tabular}{cccc}
\toprule  
Response Dimension       & $c_k$ (Task 3) & $c_k$ (Task 3) & $r_k$ (Task 4) \\
\midrule
$2$ & $-0.5$ & $2$ & $3$ \\
$5$ & $-0.8$ & $2$ & $4$ \\
$10$ & $1.1$ & $2$ & $5.5$ \\
$30$ & $1.6$ & $2$ & $9.5$ \\
\bottomrule
\end{tabular}
\end{table}

Table~\ref{tab:fdp34} and Table~\ref{tab:power34} summarize the observed power and FDR for the two additional tasks respectively, where the dimension is $d=30$ and nominal FDR level is $q=0.3$.
For the nonconvex tasks, \texttt{mCS-dist} demonstrates inferior performance compared to the baseline \texttt{bi}. In contrast, \texttt{mCS-learn} performs comparably to \texttt{bi} in relatively easy settings and surpasses \texttt{bi} in more challenging scenarios, such as Settings 3 and 6.
Consequently, \texttt{mCS-learn} emerges as the preferred choice when $R$ is nonconvex or otherwise irregular.

\begin{table}[!h]
\centering
\small
\setlength{\tabcolsep}{12pt}
\renewcommand{\arraystretch}{1.2}
\caption{Observed FDR of different methods for Task 3 and 4.}

\label{tab:fdp34}
\begin{tabular}{lcccccc}
\toprule
 & \multicolumn{3}{c}{Task 3} & \multicolumn{3}{c}{Task 4} \\
\cmidrule(lr){2-4} \cmidrule(lr){5-7}
Setting & \texttt{bi} & \texttt{mCS-d} & \texttt{mCS-l} & \texttt{bi} & \texttt{mCS-d} & \texttt{mCS-l} \\
\midrule
1 & $0.321$ & $0.267$ & $0.280$ & $0.327$ & $0.298$ & $0.304$ \\
2 & $0.312$ & $0.228$ & $0.332$ & $0.381$ & $0.271$ & $0.286$ \\
3 & $0.334$ & $0.257$ & $0.276$ & $0.385$ & $0.276$ & $0.275$ \\
4 & $0.314$ & $0.328$ & $0.277$ & $0.298$ & $0.298$ & $0.284$ \\
5 & $0.273$ & $0.266$ & $0.279$ & $0.305$ & $0.288$ & $0.313$ \\
6 & $0.237$ & $0.255$ & $0.303$ & $0.335$ & $0.263$ & $0.265$ \\
\bottomrule
\end{tabular}
\end{table}

\begin{table}[!h]
\centering
\small
\setlength{\tabcolsep}{12pt}
\renewcommand{\arraystretch}{1.2}
\caption{Observed power of different methods for Task 3 and 4.}

\label{tab:power34}
\begin{tabular}{lcccccc}
\toprule
 & \multicolumn{3}{c}{Task 3} & \multicolumn{3}{c}{Task 4} \\
\cmidrule(lr){2-4} \cmidrule(lr){5-7}
Setting & \texttt{bi} & \texttt{mCS-d} & \texttt{mCS-l} & \texttt{bi} & \texttt{mCS-d} & \texttt{mCS-l} \\
\midrule
1 & $\color{red}{0.305}$ & $0.052$ & $0.255$ & $\color{red}{0.763}$ & $0.311$ & $0.555$ \\
2 & $0.008$ & $0.011$ & $\color{red}{0.017}$ & $\color{red}{0.098}$ & $0.015$ & $0.089$ \\
3 & $\color{red}{0.018}$ & $0.005$ & $0.017$ & $0.043$ & $0.012$ & $\color{red}{0.044}$ \\
4 & $\color{red}{0.307}$ & $0.050$ & $0.241$ & $\color{red}{0.589}$ & $0.233$ & $0.446$ \\
5 & $0.024$ & $0.001$ & $\color{red}{0.034}$ & $0.114$ & $0.017$ & $\color{red}{0.117}$ \\
6 & $0.023$ & $0.001$ & $\color{red}{0.040}$ & $0.036$ & $0.016$ & $\color{red}{0.062}$ \\
\bottomrule
\end{tabular}
\end{table}

\subsubsection{Other Pretrained Models $\hat\mu$}
\label{other_mu}
We also test the performance of various methods when the underlying model $\hat\mu$ is less predictive. To simulate this effect, we train $\hat\mu$ as a linear model instead of support vector machines as in Section~\ref{sec:simu}. All other setups remain unchanged.

Table~\ref{tab:linear_mu_fdp} and Table~\ref{tab:linear_mu_power} summarize the FDR and power of different methods when $\hat\mu$ is a fitted linear model.
Although overall performance declines for most procedures, \texttt{mCS-dist} and \texttt{mCS-learn} remain comparatively stronger than the baselines. In particular, \texttt{mCS-learn} appears less sensitive to the imperfect linear fit, thanks to its data-adaptive nonconformity score $V^\theta$, which is learned based on the model $\hat\mu$ and the observed data.

\begin{table}[!h]
\centering
\small
\setlength{\tabcolsep}{12pt}
\renewcommand{\arraystretch}{1.2}
\caption{Observed FDR of different methods for lower response dimensions.}
\label{tab:linear_mu_fdp}

\begin{tabular}{lccccccccc}
\toprule
 & \multicolumn{6}{c}{Task 1} & \multicolumn{3}{c}{Task 2} \\
\cmidrule(lr){2-7} \cmidrule(lr){8-10}
Setting & \texttt{CS\_int} & \texttt{CS\_ib} & \texttt{CS\_is} & \texttt{bi} & \texttt{mCS-d} & \texttt{mCS-l} & \texttt{bi} & \texttt{mCS-d} & \texttt{mCS-l} \\
\midrule
1 & $0.776$ & $0.101$ & $0.266$ & $0.274$ & $0.264$ & $0.267$ & $0.257$ & $0.273$ & $0.276$\\
2 & $0.798$ & $0.000$ & $0.230$ & $0.269$ & $0.299$ & $0.288$ & $0.245$ & $0.222$ & $0.275$ \\
3 & $0.742$ & $0.000$ & $0.210$ & $0.312$ & $0.235$ & $0.206$ & $0.232$ & $0.223$ & $0.287$\\
4 & $0.803$ & $0.000$ & $0.267$ & $0.318$ & $0.290$ & $0.282$ & $0.307$ & $0.286$ & $0.260$\\
5 & $0.813$ & $0.000$ & $0.174$ & $0.306$ & $0.279$ & $0.272$ & $0.273$ & $0.298$ & $0.285$\\
6 & $0.779$ & $0.000$ & $0.268$ & $0.259$ & $0.306$ & $0.264$ & $0.292$ & $0.267$ & $0.279$\\
\bottomrule
\end{tabular}
\end{table}

\begin{table}[!h]
\centering
\small
\setlength{\tabcolsep}{12pt}
\renewcommand{\arraystretch}{1.2}
\caption{Observed power of different methods for lower response dimensions.}
\label{tab:linear_mu_power}

\begin{tabular}{lccccccccc}
\toprule
 & \multicolumn{6}{c}{Task 1} & \multicolumn{3}{c}{Task 2} \\
\cmidrule(lr){2-7} \cmidrule(lr){8-10}
Setting & \texttt{CS\_int} & \texttt{CS\_ib} & \texttt{CS\_is} & \texttt{bi} & \texttt{mCS-d} & \texttt{mCS-l} & \texttt{bi} & \texttt{mCS-d} & \texttt{mCS-l} \\
\midrule
1 & $1.000$ & $0.204$ & $0.467$ & $0.195$ & ${\color{red}0.550}$ & $0.379$ & $0.063$ & ${\color{red}0.842}$ & $0.583$\\
2 & $1.000$ & $0.000$ & $0.024$ & $0.077$ & $0.071$ & ${\color{red}0.089}$ & $0.386$ & $0.337$ & ${\color{red}0.419}$ \\
3 & $1.000$ & $0.000$ & $0.026$ & $0.066$ & $0.074$ & ${\color{red}0.079}$ & $0.129$ & $0.138$ & ${\color{red}0.148}$\\
4 & $1.000$ & $0.000$ & $0.203$ & $0.120$ & ${\color{red}0.309}$ & $0.186$ & $0.040$ & ${\color{red}0.408}$ & $0.188$\\
5 & $1.000$ & $0.000$ & $0.016$ & $0.036$ & $0.048$ & ${\color{red}0.051}$ & ${\color{red}0.176}$ & $0.135$ & $0.162$\\
6 & $1.000$ & $0.000$ & $0.028$ & $0.042$ & $0.039$ & ${\color{red}0.048}$ & $0.056$ & $0.056$ & ${\color{red}0.073}$\\
\bottomrule
\end{tabular}
\end{table}

\subsubsection{Comparison of Loss Functions in \textnormal{\texttt{mCS-learn}}} 
\label{extra_experiments:loss_comp}

Here we compare the performance of \texttt{mCS-learn} with loss $L_1$ \eqref{eq:loss} and loss $L_2$ \eqref{eq:alt_loss}. The response dimension is $d=30$, and the nominal FDR level is $q=0.3$.

\begin{table}[!h]
\centering
\small
\setlength{\tabcolsep}{12pt}
\renewcommand{\arraystretch}{1.2}
\caption{Observed FDR of \texttt{mCS-learn} with loss $L_1$ and $L_2$.}

\begin{tabular}{lcccc}
\toprule
 & \multicolumn{2}{c}{Task 1} & \multicolumn{2}{c}{Task 2} \\
\cmidrule(lr){2-3} \cmidrule(lr){4-5}
Setting & $L_1$ & $L_2$ & $L_1$ & $L_2$ \\
\midrule
1 & $0.255$ & $0.251$ & $0.276$ & $0.279$ \\
2 & $0.279$ & $0.267$ & $0.278$ & $0.273$ \\
3 & $0.299$ & $0.278$ & $0.331$ & $0.263$ \\
4 & $0.274$ & $0.315$ & $0.280$ & $0.254$ \\
5 & $0.210$ & $0.239$ & $0.278$ & $0.273$ \\
6 & $0.296$ & $0.258$ & $0.361$ & $0.207$ \\
\bottomrule
\end{tabular}
\end{table}

\begin{table}[!h]
\centering
\small
\setlength{\tabcolsep}{12pt}
\renewcommand{\arraystretch}{1.2}
\caption{Observed power of \texttt{mCS-learn} with loss $L_1$ and $L_2$.}
\label{tab:losses}
\begin{tabular}{lcccc}
\toprule
 & \multicolumn{2}{c}{Task 1} & \multicolumn{2}{c}{Task 2} \\
\cmidrule(lr){2-3} \cmidrule(lr){4-5}
Setting & $L_1$ & $L_2$ & $L_1$ & $L_2$ \\
\midrule
1 & $0.040$ & $\color{red}{0.325}$ & $0.051$ & $\color{red}{0.534}$ \\
2 & $0.024$ & $\color{red}{0.109}$ & $0.111$ & $\color{red}{0.421}$ \\
3 & $0.022$ & $\color{red}{0.102}$ & $0.041$ & $\color{red}{0.179}$ \\
4 & $0.034$ & $\color{red}{0.199}$ & $0.019$ & $\color{red}{0.180}$ \\
5 & $0.009$ & $\color{red}{0.042}$ & $0.041$ & $\color{red}{0.189}$ \\
6 & $0.012$ & $\color{red}{0.034}$ & $0.032$ & $\color{red}{0.061}$ \\
\bottomrule
\end{tabular}
\end{table}

As noted in the main article, employing the $L_2$ loss obviates the need for an additional round of smooth ranking, thereby enhancing both numerical stability and training efficacy. This explains the superior power observed in Table~\ref{tab:losses} when \texttt{mCS-learn} utilizes the $L_2$ loss.

\subsubsection{Comparison of Score Families in \textnormal{\texttt{mCS-learn}}}

When learning $f_\theta$ in \texttt{mCS-learn}, various forms of data may be utilized as inputs, giving rise to different score families. 
The simplest approach draws exclusively on the covariates $\boldsymbol{x}$ as features, making it applicable even in scenarios where the model $\hat\mu$ is not available. Alternatively, incorporating the response $\boldsymbol{y}$ expands the family to~\eqref{eq:score_class}. When $\hat{\mu}$ is available, its predictions $\hat{\mu}(\boldsymbol{x})$ can also be included as input. Although this does not increase the overall expressiveness of the family, it can accelerate the training process.
In this section, we evaluate the performance of the following score families with varying complexity:
\begin{align*}
    1. \quad &\textnormal{Covariate only: } M \cdot\mathbbm{1}\{ \boldsymbol{y} \notin R^c \cup \partial R\} - f_\theta(\boldsymbol{x}; R). \\
    2. \quad &\textnormal{Prediction only: } M \cdot\mathbbm{1}\{ \boldsymbol{y} \notin R^c \cup \partial R\} - f_\theta(\hat\mu(\boldsymbol{x}); R). \\
    3. \quad &\textnormal{Covariate and Prediction: } M \cdot\mathbbm{1}\{ \boldsymbol{y} \notin R^c \cup \partial R\} - f_\theta(\boldsymbol{x}, \hat\mu(\boldsymbol{x}); R). \\
    4. \quad &\textnormal{Full Family \eqref{eq:family}}: M \cdot\mathbbm{1}\{ \boldsymbol{y} \notin R^c \cup \partial R\} - f_\theta(\boldsymbol{x}, \boldsymbol{y}; R). \\
    5. \quad &\textnormal{All available information: } M \cdot\mathbbm{1}\{ \boldsymbol{y} \notin R^c \cup \partial R\} - f_\theta(\boldsymbol{x}, \hat\mu(\boldsymbol{x}), \boldsymbol{y}; R).
\end{align*}
As discussed in the main article, families 4 and 5 incorporate $\boldsymbol{y}$ as an input, which compromises the regional monotonicity of the score unless $M$ is sufficiently large. 
Table~\ref{tab:score_families_fdppower} summarizes the FDR and power of \texttt{mCS-learn} with different families, with Setting 3. For both Task 1 and Task 2, families 4 and 5 violated FDR control when $M = 10^{3}$. Among the three subfamilies that maintained valid FDR control, family 3 exhibited the best performance.



\begin{table}[!h]
\centering
\small
\setlength{\tabcolsep}{12pt}
\renewcommand{\arraystretch}{1.2}
\caption{Observed FDR and power of \texttt{mCS-learn} with different score families.}
\label{tab:score_families_fdppower}

\begin{tabular}{lcccc}
\toprule
& \multicolumn{2}{c}{FDR} & \multicolumn{2}{c}{Power} \\
\cmidrule(lr){2-3} \cmidrule(lr){4-5}
Family & Task 1 & Task 2 & Task 1 & Task 2 \\
\midrule
1 & $0.298$ & $0.237$ & ${\color{red}0.108}$ & $0.165$ \\
2 & $0.260$ & $0.291$ & $0.096$ & $0.175$ \\
3 & $0.278$ & $0.263$ & $0.102$ & ${\color{red}0.179}$ \\
4 & $0.711$ & $0.594$ & $0.972$ & $0.807$ \\
5 & $0.667$ & $0.494$ & $0.782$ & $0.536$\\
\bottomrule
\end{tabular}
\end{table}

\section{Additional Details for Real Data Application (Section~\ref{sec:real_data})} \label{real_results}

\subsection{Overview of Drug Discovery Data and Configuration of the Selection Tasks}

The drug discovery dataset we used in Section~\ref{sec:real_data} is compiled from various public sources \cite{wenzel2019predictive, iwata2022predicting, kim2023pubchem, watanabe2018predicting, falcon2022reliable, esposito2020combining, braga2015pred, aliagas2022comparison, perryman2020pruned, meng2022boosting, vermeire2022predicting}.
Because the integrated data contained missing values, we employed Chemprop \cite{yang2019analyzing, heid2023chemprop} to impute these entries. The resulting imputed dataset was then used in all subsequent experiments.
The processed dataset contains $n=22805$ data points.

We list the names, units and cutoffs of the responses (only relevant to the first selection task) in the imputed dataset in Table~\ref{tab:drug_data_responses}, and provide detailed descriptions of their biological significance and drug discovery relevance in Figure~\ref{fig:desc}.
Recall that in the first task we consider target regions of the shape $R = \{\boldsymbol{y}: y_k \geq c_k \hspace{0.2cm} \forall k\}$, 
where the cutoffs are the values $c_k$ defining the selection problem.  Figure~\ref{fig:hist} shows the distribution of these 15 responses, with vertical red lines indicating their corresponding cutoffs. Approximately 21\% of the test dataset compounds exceed all 15 thresholds, thereby qualifying for selection. 

For the second task, the target region is defined as a sphere $\{\boldsymbol{y}: \|\boldsymbol{y} - \boldsymbol{c}\|_2 \leq r \}$. For convenience, we take the center of the sphere the same as the cutoffs $c_k$ in task 1, and let $r = 2.4$. Under this definition, approximately 24\% of the compounds qualify for selection.
Similarly, for the third task where the target region is the complement of a sphere $\{\boldsymbol{y}: \|\boldsymbol{y} - \boldsymbol{c}\|_2 \leq r' \}$, we adopt the same center $\boldsymbol{c}$ and set $r' = 3.4$. This choice ensures that 18\% of the compounds qualify for selection.

\begin{table}[!h]
\centering
\small
\setlength{\tabcolsep}{12pt}
\renewcommand{\arraystretch}{1.4}
\caption{List of responses in the drug discovery dataset.}
\label{tab:drug_data_responses}

\begin{tabular}{lll}
\toprule
Name & Unit & Cutoff $c_k$ \\
\midrule
CL\_microsome\_human & $\log10$ (mL/min/kg) & $4$\\
CL\_microsome\_mouse & $\log10$ (mL/min/kg) & $4$\\
CL\_microsome\_rat & $\log10$ (mL/min/kg) & $4$\\
CL\_total\_dog & $\log10$ (mL/min/kg) & $0.5$\\
CL\_total\_human & $\log10$ (mL/min/kg) & $0$\\ 
CL\_total\_monkey & $\log10$ (mL/min/kg) & $0.5$\\
CL\_total\_rat & $\log10$ (mL/min/kg) & $1$\\
CYP2C8\_inhibition & $\log10$ (nMolar IC50) & $3.5$\\
CYP2C9\_inhibition & $\log10$ (nMolar IC50) & $3.5$\\
CYP2D6\_inhibition & $\log10$ (nMolar IC50) & $3.5$\\
CYP3A4\_inhibition & $\log10$ (nMolar IC50) & $3.5$\\
Papp\_Caco2 & $\log10$ ($10^{-6}$cm/s) & $0.8$\\
Pgp\_human & $\log10$ (efflux ratio) & $-0.2$\\
hERG\_binding & $\log10$ (nMolar IC50) & $3.5 $\\ 
LogD\_pH\_7.4 & $\log10$ (M/M) & $2$ \\ 
\bottomrule 
\end{tabular}
\end{table}

\newcommand{\desctoprule}{\noindent\rule{\textwidth}{0.5pt}\vspace{0.5em}}
\newcommand{\descbottomrule}{\vspace{0.5em}\noindent\rule{\textwidth}{0.5pt}}

\begin{figure}[!h]
    \centering
    \includegraphics[width=0.93\textwidth]{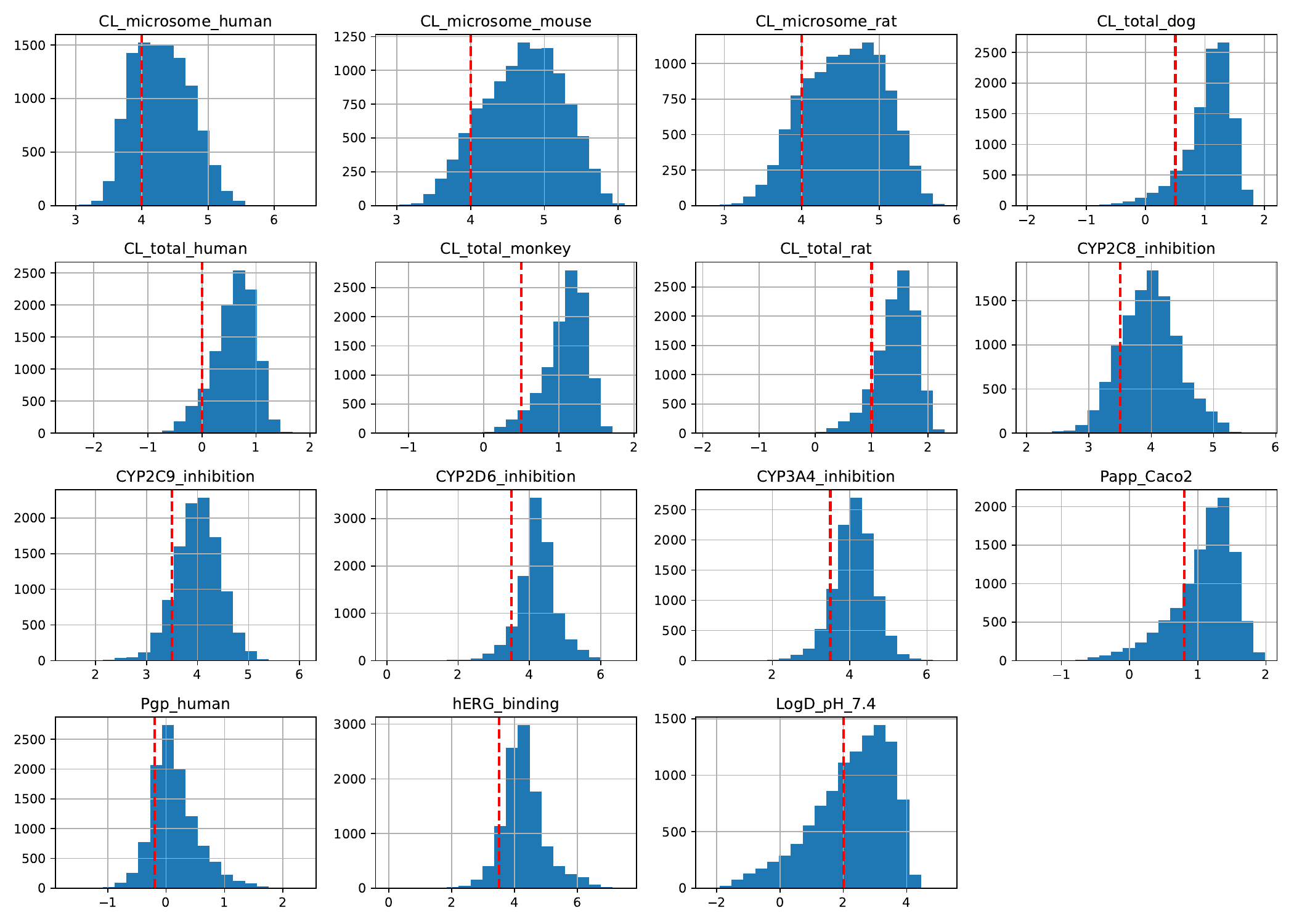}  
    \caption{Histograms of 15 responses in the drug dataset. The vertical red lines denote the corresponding cutoffs (for the first task) for each response. }
    \label{fig:hist}  
\end{figure}

\clearpage

\begin{figure}[!h]
\centering
\begin{minipage}{0.9\textwidth}
\caption{Description and drug discovery relevance of the responses.}
\label{fig:desc}
\desctoprule
\begin{description}[align=left, labelwidth=3.2cm, leftmargin=!, font=\normalfont, style=nextline, topsep=0pt, itemsep=0pt]
  \item[\textbf{CL\_microsome\_human}:] 
      Intrinsic metabolic clearance in human liver microsomes. (Help predict rate of liver metabolism in human through the in vitro in vivo correlation.)
  \item[\textbf{CL\_microsome\_mouse}:] 
      Intrinsic metabolic clearance in mouse liver microsomes. (Help predict rate of liver metabolism in human through the in vitro in vivo correlation.)
  \item[\textbf{CL\_microsome\_rat}:] 
      Intrinsic metabolic clearance in rat liver microsomes. (Help predict rate of liver metabolism in human through the in vitro in vivo correlation.)
  \item[\textbf{CL\_total\_dog}:] 
      Total body clearance measured in vivo in dogs. (Drug exposure in this species, crucial for determine human dose regimens through translational modeling.)
  \item[\textbf{CL\_total\_human}:] 
      Total body clearance measured in vivo in humans  (Drug exposure in this species, crucial for determine human dose regimens through translational modeling.)
  \item[\textbf{CL\_total\_monkey}:] 
      Total body clearance measured in vivo in monkeys. (Drug exposure in this species, crucial for determine human dose regimens through translational modeling.)
  \item[\textbf{CL\_total\_rat}:] 
      Total body clearance measured in vivo in rats.  (Drug exposure in this species, crucial for determine human dose regimens through translational modeling.)
  \item[\textbf{CYP2C8\_inhibition}:] 
      Inhibition potential against human CYP2C8 enzyme. (Assessment of risk of drug-drug interactions (DDIs) involving drugs metabolized by CYP2C8. Important for safety assessment.)
  \item[\textbf{CYP2C9\_inhibition}:] 
      Inhibition potential against human CYP2C9 enzyme. (Assessment of drug-drug interactions (DDIs) involving drugs metabolized by CYP2C9 (e.g., warfarin). Important for safety assessment.)
  \item[\textbf{CYP2D6\_inhibition}:] 
      Inhibition potential against human CYP2D6 enzyme. (Assessment of drug-drug interactions (DDIs) involving drugs metabolized by CYP2D6 (e.g., some antidepressants, beta-blockers). Important for safety assessment.)
  \item[\textbf{CYP3A4\_inhibition}:] 
      Inhibition potential against human CYP3A4 enzyme.  (Assessment of drug-drug interactions (DDIs) involving drugs metabolized by CYP3A4 (a very common pathway). Important for safety assessment.)
  \item[\textbf{Papp\_Caco2}:] 
      Apparent permeability across Caco-2 cell monolayer. (Potential for oral absorption by crossing the intestinal wall. High Papp suggests good absorption is likely.)
  \item[\textbf{Pgp\_human}:] 
      Interaction with human P-glycoprotein (Pgp) efflux transporter. (Assess if the drug is pumped out of cells by Pgp, potentially limiting oral absorption and brain penetration.)
  \item[\textbf{hERG\_binding}:] 
      Binding/inhibition of the hERG potassium channel. (Assessment of risk of cardiac toxicity (QT prolongation, arrhythmias). A critical safety screen; high binding is a major safety concern.)
  \item[\textbf{LogD\_pH\_7.4}:] 
      Logarithm of the distribution coefficient at pH 7.4. (Predicts drug's lipophilicity (fat vs. water solubility) under physiological conditions. Influences absorption, distribution, membrane crossing, and CNS penetration.)
\end{description}
\descbottomrule
\end{minipage}
\end{figure}

\clearpage

\subsection{Extra Experiments on Real Data}

\subsubsection{Numerical Stability of Different Methods}

In this section, we test the numerical stability of \texttt{mCS-dist} and \texttt{mCS-learn} on real data. To do this, unlike previous experiments where we sample new data or randomly partition data for each iteration, we fix the pretrained model $\hat\mu$, the calibration data $\mathcal{D}_{\textnormal{cal}}$ and the test data $\mathcal{D}_{\textnormal{test}}$. 
This is to avoid the variability from data sampling or pretraining, and only consider the inherent variance of the methods. 
Due to the high computation cost of retraining $f_\theta$, we also keep $f_\theta$ the same across different iterations.
We keep all other setups and configurations unchanged as in Section~\ref{sec:real_data}.

\setlength{\tabcolsep}{12pt}
\renewcommand{\arraystretch}{1.2}
\begin{table}[!h]
\centering
\small
\caption{Observed standard error of FDRs for different methods with real data.}

\label{tab:stab_fdp}
\begin{adjustbox}{max width=0.7\textwidth}
\begin{tabular}{lccccccc}
\toprule

Task & $q$ & \texttt{CS\_int} & \texttt{CS\_ib} & \texttt{CS\_is} & \texttt{bi} & \texttt{mCS-d} & \texttt{mCS-l}\\
\midrule
1 & $0.3$ & $0.000$ & $0.000$ & $0.248$ & $0.000$ & $0.000$ & $0.000$ \\
2 & $0.3$ & $-$ & $-$ & $-$ & $0.000$ & $0.000$ & $0.002$ \\
3 & $0.3$ & $-$ & $-$ & $-$ & $0.000$ & $0.021$ & $0.004$ \\
\midrule
1 & $0.5$ & $0.000$ & $0.000$ & $0.032$ & $0.000$ & $0.001$ & $0.004$ \\
2 & $0.5$ & $-$ & $-$ & $-$ & $0.000$ & $0.000$ & $0.002$ \\
3 & $0.5$ & $-$ & $-$ & $-$ & $0.000$ & $0.028$ & $0.005$ \\
\bottomrule
\end{tabular}
\end{adjustbox}
\end{table}

\begin{table}[!h]
\centering
\small
\caption{Observed standard error of powers for different methods with real data.}

\label{tab:stab_power}
\begin{adjustbox}{max width=0.7\textwidth}
\begin{tabular}{lccccccc}
\toprule

Task & $q$ & \texttt{CS\_int} & \texttt{CS\_ib} & \texttt{CS\_is} & \texttt{bi} & \texttt{mCS-d} & \texttt{mCS-l}\\
\midrule
1 & $0.3$ & $0.000$ & $0.000$ & $0.022$ & $0.000$ & $0.000$ & $0.000$ \\
2 & $0.3$ & $-$ & $-$ & $-$ & $0.000$ & $0.000$ & $0.001$ \\
3 & $0.3$ & $-$ & $-$ & $-$ & $0.000$ & $0.005$ & $0.002$ \\
\midrule
1 & $0.5$ & $0.000$ & $0.000$ & $0.056$ & $0.000$ & $0.000$ & $0.004$ \\
2 & $0.5$ & $-$ & $-$ & $-$ & $0.000$ & $0.000$ & $0.004$ \\
3 & $0.5$ & $-$ & $-$ & $-$ & $0.000$ & $0.011$ & $0.005$ \\
\bottomrule
\end{tabular}
\end{adjustbox}
\end{table}

For all methods, the standard errors for both FDR and power remain low, with the exception of \texttt{CS\_is} at a nominal level of $q=0.3$. This phenomenon likely arises from the unstable nature of the subroutine nominal level search at lower user-specified nominal levels. The failure of \texttt{CS\_is} to control the FDR at this low nominal level (Figure~\ref{fig:varying_q}) further corroborates this interpretation. 

\end{document}